\documentclass[11pt,letterpaper]{article}
\usepackage{natbib}

\usepackage[margin=1.1in]{geometry}
\usepackage[utf8]{inputenc}
\usepackage{amsfonts}
\usepackage{mathrsfs}
\usepackage{graphicx}
\usepackage{wrapfig}
\usepackage{pgfplots}
\usepackage{environ}
\pgfplotsset{compat=1.4}
\usepackage{amsmath}
\usepackage{amsthm}
\usepackage{amssymb}
\usepackage{amsfonts}
\usepackage{mathrsfs}
\usepackage{aliascnt}
\usepackage{natbib}
\usepackage{color}
\usepackage{tikz}
\usetikzlibrary{patterns}
\usepackage{subfig}
\usepackage{floatrow}
\usepackage{sidecap}
\usepackage{mdwlist}
\usepackage{bbm}
\usepackage{authblk}
\usepackage{bbm}

\usepackage[capitalize]{cleveref}

\newtheorem{definition}{Definition}
\newtheorem{lemma}[definition]{Lemma}
\newtheorem{corollary}[definition]{Corollary}

\newtheorem{theorem}[definition]{Theorem}

\newtheorem{proposition}[definition]{Proposition}

\theoremstyle{definition}
\newtheorem{example}[definition]{Example}

\newcommand{\AutoAdjust}[3]{\mathchoice{ \left #1 #2  \right #3}{#1 #2 #3}{#1 #2 #3}{#1 #2 #3} }
\newcommand{\Xcomment}[1]{{}}

\newcommand{\InBrackets}[1]{\AutoAdjust{[}{#1}{]}}
\newcommand{\Ex}[2][]{\operatorname{\mathbb E}_{#1}\InBrackets{#2}}
\newcommand{\Prx}[2][]{\operatorname{\text{Pr}}_{#1}\InBrackets{#2}}
\def\prob{\Prx}
\def\expect{\Ex}

\newcommand{\OPT}{*}

\newcommand{\agent}{i}
\newcommand{\agentiter}{\agent}
\newcommand{\numagents}{n}
\newcommand{\agentalt}{\agent^{\prime}}

\newcommand{\action}{a}
\newcommand{\actions}{\mathbf{\action}}
\newcommand{\actionagent}{\action_{\agentiter}}
\newcommand{\actionothers}{\actions_{-\agentiter}}
\newcommand{\actspace}{A}
\newcommand{\actspaceagent}{\actspace_{\agentiter}}

\newcommand{\actspaces}{\mathbf A}

\newcommand{\actionsitem}{\actions^\mitem}

\newcommand{\actiondist}{G}

\newcommand{\actionsdist}{\mathbf{\actiondist}}

\newcommand{\feasible}{\mathcal X}
\newcommand{\allocation}{x}
\newcommand{\altallocation}{y}
\newcommand{\alloc}{\mathbf \allocation}
\newcommand{\allocopt}{\alloc^{\OPT}}
\newcommand{\allocalt}{\mathbf \altallocation}
\newcommand{\allocagent}[1][\agentiter]{\allocation_{#1}}
\newcommand{\allocoptagent}{\allocation^{\OPT}_{\agentiter}}

\newcommand{\agentallocalt}{\altallocation}
\newcommand{\allocaltagent}[1][\agent]{\agentallocalt_{#1}}
\newcommand{\allocaltagentitem}{\agentallocalt_{\agent}^{\mitem}}

\newcommand{\bidallocation}{\tilde \allocation}
\newcommand{\bidalloc}{\mathbf\bidallocation}

\newcommand{\bidallocagent}{\bidallocation_\agent}

\newcommand{\bidallocr}{\bidalloc^\reserves}

\newcommand{\slotprob}{\alpha}

\newcommand{\perm}{\pi}

\newcommand{\bid}{b}
\newcommand{\bids}{\mathbf{\bid}}
\newcommand{\bidagent}[1][\agent]{\bid_{#1}}
\newcommand{\bidothers}[1][\agent]{\bids_{-#1}}

\newcommand{\bidcdf}{\mathcal B}
\newcommand{\bidcdfagent}{{\bidcdf}_\agent}
\newcommand{\bidcdfs}{\mathcal B}

\newcommand{\bidsalt}{\bids'}
\newcommand{\bidaltagent}{\bid_\agent'}

\newcommand{\bidpostresagent}{\tilde{\bid}_\agent}
\newcommand{\bidspostres}{\tilde{\bids}}

\newcommand{\deviation}{d}
\newcommand{\dev}{\deviation}

\newcommand{\mech}{M}

\newcommand{\mechr}[1][\reserves]{\mech^{#1}}

\newcommand{\mechanism}[1]{\textsc{#1}}

\newcommand{\AP}{\mechanism{AP}}
\newcommand{\WPB}{\mechanism{WPB}}

\newcommand{\OPTmech}{\mechanism{Opt}}

\newcommand{\pay}{p}

\newcommand{\bidpay}{\tilde\pay}
\newcommand{\bidpayment}{\mathbf{\bidpay}}
\newcommand{\bidpaymentagent}{\bidpay_{\agentiter}}

\newcommand{\WEL}{\textsc{Welfare}}
\newcommand{\REV}{\textsc{Rev}}

\newcommand{\rev}{\REV}
\newcommand{\revpar}{\mu}

\newcommand{\threshold}{\hat \bid} 
\newcommand{\thresholdagent}[1][\agentiter]{\threshold_{#1}}

\newcommand{\expthreshold}{T}

\newcommand{\expthresholdagent}{\expthreshold_\agent}

\newcommand{\interthresh}[2]{t_{#1}(#2)}

\newcommand{\cumuthreshagent}[3]{T_{#1}^{#3}(#2)}

\newcommand{\reservem}{\reserve^*}

\newcommand{\threshexpectedallpay}[2]{\expthresholdagent^{\AP}}					
\newcommand{\threshexpectedwinnerpaysbid}[2]{\expthresholdagent^{\WPB}}					

\newcommand{\threshpt}[1]{\thresholdagent(#1)}							
\newcommand{\threshexpected}[2]{\expthresholdagent}					

\newcommand{\threshbound}{\underline{\expthreshold}}						
\newcommand{\threshboundagent}{\threshbound_{\agentiter}}				
\newcommand{\threshboundexpected}[2]{\threshboundagent}			

\newcommand{\equivbid}{\beta}
\newcommand{\equivbidagent}{\equivbid_\agent}

\newcommand{\equivbidagentitem}{\equivbidagent^\mitem}

\newcommand{\equivbidact}[1]{\equivbidagent(#1)}

\newcommand{\util}{u}
\newcommand{\utilagent}{\util_{\agentiter}}

\newcommand{\bidutil}{\tilde u}
\newcommand{\bidutilagent}{\bidutil_{\agentiter}}

\newcommand{\val}{v}
\newcommand{\vals}{\mathbf \val}
\newcommand{\valothers}{\vals_{-\agentiter}}
\newcommand{\valagent}[1][\agentiter]{v_{#1}}

\newcommand{\valueset}{V}
\newcommand{\values}{\mathbf\valueset}

\newcommand{\valuecdf}{\mathcal F_{\agent}}
\newcommand{\valuepdf}{f_{\agent}}

\newcommand{\dist}{\mathcal F}

\newcommand{\valuecdfs}{\dist}

\newcommand{\vval}{\phi}
\newcommand{\vvalagent}{\vval_{\agentiter}}

\newcommand{\mitem}{j}
\newcommand{\mechitem}{\mech^\mitem}

\newcommand{\actionagentitem}{\actionagent^\mitem}

\newcommand{\wants}{\mathcal S}
\newcommand{\wantsagent}[1][\agent]{\wants_{#1}}

\newcommand{\numitems}{m}

\newcommand{\reserve}{r}
\newcommand{\reserves}{\mathbf{\reserve}}
\newcommand{\zeros}{\mathbf{0}}
\newcommand{\reserveagent}{\reserve_{\agent}}

\newcommand{\reserveagentm}{\reserveagent^{\OPT}}
\newcommand{\reservesm}{\mathbf{\reserve}^{\OPT}}
\newcommand{\reservesothers}{\reserves_{-\agentiter}}

\newcommand{\rcr}{\revpar}
\newcommand{\res}{\eta}

\newcommand{\equivthresh}[2]{\tau_{#1}(#2)}
\newcommand{\equivthreshres}[3]{\tau_{#1}^{#3}(#2)}

\newcommand{\equivcumulativesimple}{\mathscr T}

\newcommand{\equivcumulative}[3]{{\equivcumulativesimple}_{#1}^{#3}(#2)}

\newcommand{\bidallocagentitem}{\bidallocagent^\mitem}
\newcommand{\bidallocitem}{\bidalloc^\mitem}
\newcommand{\bidpaymentagentitem}{\bidpay_{\agent}^{\mitem}}
\newcommand{\bidpaymentitem}{\bidpayment^\mitem}

\newcommand{\convcomb}{\theta}

\newcommand{\mf}{auction environment}
\newcommand{\convcombdist}{\omega}

\newcommand{\cumuthresh}{T}

\newcommand{\singlecdf}{F}
\newcommand{\reservesingm}{\reserve^{\OPT}}
\newcommand{\comp}{\beta}
\newcommand{\compagent}{\comp_\agent}
\newcommand{\wpballoc}{\hat\allocation}

\newcommand{\jointcdf}{\mathcal G}
\newcommand{\info}{\theta}
\newcommand{\mechdist}{\mathcal D}

\newcommand{\EC}[1]{}
\newcommand{\IFECELSE}[2]{#2}

\newcommand{\rc}{competitive efficiency}

\newcommand{\Rc}{Competitive efficiency}
\newcommand{\RC}{Competitive Efficiency}
\newcommand{\trc}{competitive efficiency}
\newcommand{\Trc}{Competitive efficiency}
\newcommand{\vc}{individual efficiency}
\newcommand{\Vc}{Individual efficiency}
\newcommand{\VC}{Individual Efficiency}

\newcommand{\bo}{bidding outcome}
\newcommand{\outcome}{outcome}
\newcommand{\bos}{bidding outcomes}
\newcommand{\ao}{action outcome}
\newcommand{\aos}{action outcomes}

\begin{document}

\markboth{J Hartline et al.}{Robust Analysis of Auction Equilibria}

\title{Robust Analysis of Auction Equilibria\footnote{This paper
  provides an economic interpretation on results presented in an
  extended abstract under the title ``Price of Anarchy for Auction
  Revenue'' at the fifteenth ACM Conference on Economics and
  Computation \citep{HHT-14}.  We thank Vasilis Syrgkanis for comments
  on a prior version of this paper for which simultaneous composition
  did not hold, for suggesting the study of simultaneous composition,
  and for perspective on price-of-anarchy methodology.} }

\author[1]{Jason Hartline} \author[ ]{Darrell Hoy}
\author[2]{Sam Taggart}

\affil[1]{Northwestern University}
\affil[2]{Oberlin College}

\maketitle

\begin{abstract}
  Equilibria in auctions can be very difficult to analyze, beyond the
symmetric environments where revenue equivalence renders the analysis
straightforward.  This paper takes a robust approach to evaluating the
equilibria of auctions.  Rather than identify the equilibria of an
auction under specific environmental conditions, it considers
worst-case analysis, where an auction is evaluated according to the
worst environment and worst equilibrium in that environment.  It
identifies a non-equilibrium property of auctions that governs whether
or not their worst-case equilibria are good for welfare and revenue.
This property is easy to analyze, can be refined from data, and
composes across markets where multiple auctions are run
simultaneously.
\end{abstract}

\section{Introduction}
\label{s:intro}

Equilibria in auctions can be very difficult to analyze, beyond the
symmetric environments where revenue equivalence renders the analysis
straightforward.  This paper takes a robust approach to evaluating the
equilibria of auctions.  Rather than identify the equilibria of an
auction under specific environmental conditions, it considers
worst-case analysis, where an auction is evaluated according to the
worst environment and worst equilibrium in that environment.  It
identifies a non-equilibrium property of auctions that governs whether
or not their worst-case equilibria are good for welfare and revenue.
This property is easy to analyze, can be refined from data, and
composes across markets where multiple auctions are run
simultaneously.

Classical economic analyses identify two main drivers of
inefficiency: (a) externalities and (b) incomplete information.  The analysis of this paper decomposes the performance of
auctions into two terms: {\em competitive efficiency} quantifies the
degree to which externalities cause losses in performance and {\em
  individual efficiency} quantifies the degree to which incomplete
information causes losses in performance.  

The quantity of interest for this paper is the {\em robust efficiency}
of mechanisms, i.e., the fraction of the optimal welfare or revenue
that is attained in any equilibrium and under any informational model.
By measuring the efficiency of a mechanism as the ratio of its
performance to the optimal performance we obtain bounds that are
invariant with respect to the relative magnitudes of the environment.


The analysis of the paper shows that competitive efficiency, as we
define it, is a central determinant of whether an auction is good or
bad; while individual efficiency is always relatively high.
Competitive efficiency can be low when externalities are significant
and mechanisms that are more competitively efficient reduce the impact
of externalities.  On the other hand, individual efficiency --- which
quantifies the impact of incomplete information --- is always
relatively high; though, e.g., mechanisms with the winner-pays-bid
payment format can be seen as better than those with the all-pay
payment format.  Applying the philisophy of approximation
\citep{har-13} to these two concerns: externalities are a critical
feature to be treated carefully in mechanism design while incomplete
information is more of a detail.

Introducing the notion of competitive efficiency by example, consider
any set of bids in a first-price auction.  For each agent and this
fixed set of bids, there is a minimum bid that serves as a threshold
for whether the agent wins or loses.  Consider two quantities: (a) the
revenue of the auction for the given bids and (b) the optimal revenue
of an auction if each agent instead bid their threshold.  Quantity (b)
is a measure of the level of competition in the auction.  The {\em
  competitive efficiency} $\rcr$ is the worst case over bids of the
ratio of (a) to (b), this number is at most 1 and numbers closer to 1
are more efficient.  For the first-price auction this ratio is
$\rcr=1$. The revenue is the highest bid and the thresholds of all
losers are equal to the highest bid.  Thus, the two terms, (b) the
optimal revenue from thresholds and (a) the revenue from bids, are the
same.


Even mechanisms that are competitively efficient may lose efficiency
due to the inability of the agents to precisely respond to thresholds,
i.e., due to incomplete information.  Returning to the example of
first price auctions, with deterministic equilibrium concepts like
pure Nash equilibrium, an agent's threshold is deterministic, and the
agent can respond efficiently.  On the other hand, in stochastic
equilibrium concepts like Bayes-Nash equilibrium, the threshold an
agent faces is stochastic, and the agent cannot respond efficiently.
Specifically, when the threshold is below the agent's value the agent
would prefer to bid just above the threshold; however, this threshold
is stochastic and the agent must place a single bid.  Viewing this
threshold as the seller's outside option, the inability for an agent
to respond precisely results in an additional loss of performance.
This {\em individual efficiency} $\res$ is a property of the
best-response problem faced by the agents and depends on the
equilibrium notion (e.g., pure Nash or Bayes-Nash) and payment format
(e.g., winner-pays-bid or all-pay).


This paper shows that robust efficiency can be bounded by a
combination of the individual efficiencies of the agents and
competitive efficiency of the mechanism. The main welfare analysis of
this paper proves that, broadly, the fraction of the optimal welfare
of an auction in equilibrium is at least the product of the
mechanism's competitive efficiency and the individual efficiency of
the agents' best response problems.  For example, we will show that the
individual efficiency of the agent's response for pure Nash equilibria
in winner-pays-bid mechanisms is $\res = 1$.  Combined with the
competitive efficiency of $\rcr = 1$ for the first-price auction, we
see that pure Nash equilibria in the first-price auction are fully
efficient, i.e., they obtain a $\res\rcr = 1$ fraction of the optimal
welfare.  (More precisely, pure Nash tend not to exist in the
first-price auction, but the same result approximately holds for
approximate pure Nash which exist; details given subsequently.)
Turning to revenue analysis, the main result paper shows that the
fraction of the optimal revenue of an auction with appropriate reserve
prices in equilibrium is at least half the product of its competitive
efficiency and the agents' individual efficiency.

Individual efficiency is a property of the agents' best response
problem that takes into account the incompleteness of information and
the payment format.  Thus, it is sufficient to analyze a few canonical
models of incomplete information and payment formats.  We provide the
following individual efficiency results:
\begin{itemize}
\item $\res = (1-\epsilon)$ for the winner-pays-bid payment format and
 $(1-\epsilon)$-approximate deterministic best-response
  problems (for any $\epsilon \geq 0$).  This result is applicable to
  pure $(1-\epsilon)$-Nash equilibria in deterministic auctions for
  allocating indivisible goods such as the first-price auction.  This
  result is written for approximate Nash equilibrium rather than exact
  Nash equilibrium because in winner-pays-bid auctions the latter
  tends not to exist.  
  \item $\res = 1-1/e \approx 0.63$ for the winner-pays-bid payment
    format and stochastic best-response problems.  This
    result is applicable to Bayes-Nash equilibria and, more generally, Bayes correlated
    equilibria with arbitrary information structures.  
  \item $\res = 1/2$ for the all-pay payment format and stochastic best-response problems.  This result is applicable to
    Bayes-Nash equilibria.
\end{itemize}
These results quantify the role of incomplete information in the loss
of efficiency.  When information is compete (in payoffs, actions of
opponents, and rules of the mechanism), an agent's best response
problem is fully efficient.  On the other hand we see that with
incomplete information the agent's best response problem can be
inefficient, but that inefficiency is bounded by a constant.  As
mentioned above, the best response problem of winner-pays-bid payment
formats mitigates informational inefficiencies more than that of the
all-pay payment formats.

As an example, combining these individual efficiency bounds with the competitive
efficiency of the first-price auction ($\rcr=1$), the equilibrium
welfare is at least an $\res\rcr = 0.63$ fraction of the optimal
welfare in equilibrium.  The revenue of the first-price auction with
per-agent monopoly reserves is at least a $\res\rcr/2 = 0.31$ fraction of the
optimal revenue in equilibrium.  Recall that these auctions are
optimal for welfare and revenue when the agents values are identically
distributed, the above robust welfare guarantee holds for all
correlated distributions and the above robust revenue guarantee holds
for all non-identical product distributions satisfying a standard
regularity property.  We see from this analysis that, while the
first-price auction can be inefficient, it can never be extremely
inefficient.



Competitive efficiency is a property of the rules of the auction that
map bids to winners, a.k.a., bid allocation rules.  Given the
individual efficiencies of standard best response problems of the
agents, the robust efficiency of a mechanism is approximately
governed by its competitive efficiency.  Therefore, analysis of robust
welfare and revenue of an auction approximately reduces to analysis of its
competitive efficiency.

To aid in the analysis of the competitive
efficiency we develop a number of closure properties.
\begin{itemize}
  \item Competitive efficiency is closed under reserve pricing, i.e.,
    the competitive efficiency of mechanism without reserve prices is
    equal to the competitive efficiency of the mechanism with the
    worst reserve prices.

  \item Competitive efficiency is closed under randomizations of
    mechanism, agent selection, and correlated signaling, i.e., the
    competitive efficiency of any randomized mechanism on any
    randomized sets of agents and with any information structure is
    equal to the competitive efficiency of the worst mechanism in the
    combination, set of agents, and information structure.
    As a consequence, the worst-case \rc\ is attained at degenerate randomizations over mechanisms and degenerate information structures. To characterize the \rc\ of such families of randomizations, it consequently suffices to analyze the \rc\ of simple mechanisms under full information.

  \item Competitive efficiency is closed under simultaneous
    composition (when bids are independently distributed), i.e., the
    competitive efficiency of the worst composite mechanism where a
    set of mechanisms are run in parallel is equal to the competitive
    efficiency of the worst mechanism in the set.  (In this
    composition, agents are assumed to be unit-demand, but can bid in
    multiple mechanisms at once if it is in their best interest.)
\end{itemize}
These closure properties imply that it is generally sufficient to
analyze competitive efficiency of deterministic mechanisms with
deterministic selection.  Such analyses are fairly straightforward
compared to more technically involved analyses of stochastic
equilibrium concepts.

The following mechanisms have competitive efficiency $\rcr=1$.
Intuitively, in these environments the externalities are only those of
one-for-one substitution:
\begin{itemize}
  \item Single-item multi-unit unit-demand highest-bids-win
    mechanisms, i.e., there are $k$ identical units for sale and they
    are allocated to the $k$-highest bidders. (Note that the
    first-price auction is the special case where $k=1$.)
  \item Rank-by-bid position auctions, i.e., there are $k$ positions
    with descending weights, and the highest $k$ bidders are assigned
    to these $k$ positions in order of bid.  The agent in the $j$th
    highest position receives a unit with probability equal to the
    $j$th position weight.  (The $k$-unit auction is a special case
    where the position weights are all 1.)  The position auction model
    was popularized by the studies of \citet{var-07} and
    \citet{EOS-07} of auctions for advertising on Internet search
    engines.
  \item (Single-bid) highest-bids-win matching markets, i.e., there
    are $m$ items and $n$ bidders  who each desire a single item from
    a known subset of the items.  The highest-bid-wins rule selects the
    bidders to match to maximize the sum of the matched bids.
    
  \item (Multiple-bid) per-item highest-bids-win matching markets,
    i.e., there are $m$ items and $n$ bidders who each desire a single
    item from a unknown subset of items, bids are submitted for each
    item, and the highest bidder for each item wins it.  (For this result,
    bids are required to be independently distributed in equilibrium.)
\end{itemize}

These results all follow from a single analysis of the competitive
efficiency of bid allocation rules based on greedy algorithms and the
above closure properties.  Greedy algorithms order the agents by a
function of their bids and then allocate to each one in turn, if doing
so is feasible.  For complicated
optimization problems, greedy algorithms are not generally optimal.
Our main analysis of competitive efficiency shows that the
non-strategic efficiency of the greedy algorithm, i.e., the fraction
of the optimal welfare it obtains in non-strategic environments, is
equal to its competitive efficiency. The results listed above follow
because the greedy algorithm is non-strategically efficient for these
environments. Notice that the $k$-unit highest-bids-win allocation
rule is given by the greedy-by-bid algorithm. The highest-bids-win rule for (single-bid) highest-bids-win matching markets can also be implemented greedily. The result for rank-by-bid position auctions
additionally views the position auction as a convex combination of
multi-unit auctions and invokes the closure of competitive efficiency
under convex combinations.  The result for  multiple-bid per-item matching
markets follows from the closure under simultaneous composition of the
first-price auction.

The competitive efficiency of mechanisms can be very bad when
externalities between agents are ones of many-for-one
complementarities.  The classic environment where agents are
complements is the single-minded combinatorial auction.  Here there
are $n$ agents and $m$ items and agents each desire the entirety of
known subsets of items.  The competitive efficiency of highest-bids-win
single-minded combinatorial auctions is $1/m$ and indeed
winner-pays-bid highest-bids-win auctions possess equilibria that are
$1/m$ efficient.  A classical result from the computer science
literature on algorithm design, however, shows that there is a greedy
algorithm for the single-minded combinatorial allocation problem that
is $1/\sqrt{m}$ efficient;\footnote{Computer scientists study greedy
algorithms like this one in part because computing the allocation that
maximizes the sum of winning bids is believed to be computationally
intractable.  Thus, we see that greedy algorithms should be preferred
as allocation rules for auctions both for their computational
tractability and for there competitive efficiency.}  thus the
winner-pays-bid auction with this greedy bid allocation rule is
$1/\sqrt{m}$ competitively efficient.  Explicitly designing mechanisms
to maximize competitive efficiency can significantly reduce the impact
of externalities.

\subsection{Related Work}

In symmetric environments in the canonical independent private value
model second-price, first-price, and all-pay auctions with or without
reserves are both welfare and revenue equivalent \citep{M81}.  In
asymmetric models, they are not equivalent, and their equilibria, even
under the simplifying assumption that the agents' values are uniformly
distributed, (but asymmetrically with different supports) are very
difficult to solve for \citep{KZ12}.

An approach from computer science for understanding the potentially
complex equilibria of auctions is to give robust bounds
on equilibria that do not require exactly identifying an equilibrium.
This literature analyzes the robust efficiency, a.k.a., the price of
anarchy, of games and mechanisms.  Within this literature, this paper
builds on the ``smooth games'' framework of \citet{R09} and the ``smooth
mechanisms'' extension of \citet{ST13}.  In this context, this paper
refines the smoothness framework for Bayesian games in two notable
ways. First, it decomposes smoothness into two components, separating
the consequences of best-response (individual efficiency) from the
specifics of a mechanism (competitive efficiency).  The former
quantifies losses due to incomplete information and the latter
quantifies losses due to externalities between agents.  Second, the
framework is compatible with the analysis of auction revenue by
\citet{M81} and allows for robust bounds on the equilibrium revenue of
auctions.

There are two subsequent works with strong connections to our
decomposition of robust efficiency into competitive efficiency and individual
efficiency.  First, \citet{DK15} show that competitive efficiency,
which they call ``permeability,'' is in fact a necessary condition (we
show sufficiency) for the equilibrium welfare of a mechanism to be
proven to be good via the smoothness framework.  Second, \citet{HNS15}
show how to derive empirical welfare bounds by measuring the degree to
which competitive efficiency and individual efficiency hold, without
needing to infer the agents' true values.

Recently, the economics literature has also seen a number of robust
treatments of mechanisms.  The main distinction between this
literature and the present paper is that this literature focuses on
absolute performance whereas the present paper (and the computer
science literature discussed previously) focuses on robust performance
relative to the optimal mechanism.  This difference is significant as
absolute robust analysis generally focuses on the settings were the
optimal performance is the lowest.  On the other hand, analysis
relative to the optimal performance requires that the performance is
close to what is possible absent concerns of robustness and is
invariant to the scale.  The relative framework should be preferred
when the range of the optimal performance varies significantly or
robustness is required at all scales.

The robust analyses in economics identify
mechanisms with good worse-case revenue when there is ambiguity with
respect to key aspects of the models.  For example, \citet{car-17}
considers a revenue maximizing seller with multiple items and a buyer
with values distributed with known marginals.  The correlation
structure is ambiguous.  He shows that the max-min mechanism is a linear
pricing.  \citet{BD-21} consider the design of revenue optimal
common-value auctions that are robust to information structures.

The most related works from the robust mechanism design literature give lower
bounds on the performance of standard mechanisms under permissive
assumptions on the beliefs of agents.  \citet{BBM-17} consider a
robust analysis of the revenue in a first-price auction with respect
to the knowledge of the agents.  They derive tight lower bounds on
this revenue based on a characterization of the minimum distribution
over winning bids.  In contrast, though our welfare bounds are tight,
our revenue bounds are loose.  \citet{BBM-19} give worst-case, over
information structure, revenue rankings of standard auction formats.
In contrast, the bounds of this paper are relative to the optimal
performance and apply broadly beyond single-item auctions.

Another strand of literature derives revenue guarantees for the
welfare-optimal Vickrey-Clarke-Groves (VCG) mechanism in asymmetric
environments. \citet{HR09} show that VCG with monopoly reserves, a
carefully chosen anonymous reserve, or duplicate bidders achieves
revenue that is a constant approximation to the revenue optimal
auction. \citet{DRY10} show that the single-sample mechanism,
essentially VCG using a single sample from the distribution as a
reserve, achieves approximately optimal revenue in broader
environments. \citet{RTY12} showed that in broader environments,
including matching environments, limiting the supply of items in
relation to the number of bidders gives a constant approximation to
the optimal auction.  See \citet{har-13} for a survey of results in
this area.

In the computer science literature a canonical robust mechanism design problem is that of maximizing the
ratio of the revenue from a single-item auction to the optimal auction
in worst case over valuation distributions of the agents.  This
problem was posed by \citet{DRY-15}, further considered by
\citet{FILS-15} and \citet{AB-20}, and optimally solved by
\citet{HJL-20}.  These papers aim to design optimal incentive
compatible mechanisms.  In contrast, the present paper gives robust
analyses of standard mechanisms (which are not incentive compatible).

\section{Preliminaries}
\label{sec:prelim}

This paper studies mechanisms for single-parameter agents.
A mechanism consists of action spaces $\actspaceagent$ for each agent $\agent$ (and joint action space $\actspaces = \prod_\agent
\actspaceagent$), an allocation rule $\bidalloc$, and a payment
rule $\bidpayment$.
Given the profile $\actions$ of actions
selected by each agent, the
mechanism computes an allocation level $\bidallocagent(\actions)\in[0,1]$ and payment $\bidpaymentagent(\actions)\in\mathbb R$ for each agent $\agent$, with 
$\bidalloc(\actions)$  and 
$\bidpayment(\actions)$ describing the full profiles of allocations and payments, respectively.
For indivisible goods, $\bidallocagent(\actions)$ captures agent $\agent$'s probability of service.
Each agent $\agent$ has a value $\valagent$ for service, and linear utility function $\bidutilagent^{\valagent}(\actions) = \valagent\bidallocagent(\actions)-\bidpaymentagent(\actions)$. We omit the superscript of $\valagent$ when it is clear from context.

Allocation for an $n$-agent mechanism is constrained by a {\em feasibility environment} $\feasible$.
For example, in selling a single item, $\feasible=\{\alloc\in[0,1]^{\numagents}\,|\,\sum_i \allocagent\leq 1\}$.
Settings we consider will be {\em downward-closed}, in the sense that for any $\alloc\in\feasible$ and any index $\agent$, $(0,\alloc_{-\agent})\in\feasible$.
Even fixing a mechanism's payment format (e.g.\ winner-pays-bid), rich feasibility settings admit a variety of allocation rules.
In our analysis framework, \trc\ (Section~\ref{sec:rc}) quantifies the consequences of this variation.

We allow a broad range of solution concepts within the private values model,\footnote{In particular, values are not interdependent, but may be correlated unless explicitly stated otherwise.} and take an equilibrium to be a joint distribution over actions and values.
Formally, let $\values=\valueset_1\times\ldots\times\valueset_n\subseteq\mathbb R^n_+$ be the space of $\numagents$-agent value profiles.
An equilibrium is a joint distribution $\jointcdf$ over $\values\times\actspaces$, which may be correlated across agents' actions, agents' values, or between actions and values.
Different concepts will impose different best response and independence conditions on these distributions.
We use $\bidcdfs$ and $\valuecdfs$ to denote the marginal distributions over the action profile $\actions$ and value profile $\vals$ respectively, with marginal distributions $\bidcdfagent$ and $\valuecdf$ for agent $\agent$.
Throughout, we maintain the agent-normal form interpretation of the distributions: each agent $\agent$ is selected from a population with distribution $\valuecdf$, and given selection $\vals$ the value $\valagent$ agents in population $\agent$ play according to $\bidcdf~|~\vals$.
Given mechanism $\mech$ and equilibrium $\jointcdf$, our two objectives of interest are revenue, given by $\rev(\mech,\jointcdf)=\mathbb E_{\actions}[\sum_\agent\bidpaymentagent(\actions)]$, and welfare, given by $\WEL(\mech,\jointcdf)=\mathbb E_{\vals,\actions}[\sum_\agent\bidutilagent^{\valagent}(\actions)]+\rev(\mech,\jointcdf)=\mathbb E_{\vals,\actions}[\sum_\agent\valagent\bidallocagent(\actions)]$.
Separating the bid and value distributions explicitly, we also write $\rev(\mech,\valuecdfs,\bidcdfs)$ and $\WEL(\mech,\valuecdfs,\bidcdfs)$.

The mechanisms we study are non-truthful and depend only minimally on priors.
Consequently, these mechanisms often have equilibria with suboptimal welfare or revenue.
We therefore pursue robust, or worst-case, approximation analyses, and ask how far from optimal a mechanism can be, quantified over equilibria in a family. 
Formally, given value distribution $\valuecdfs$, denote the optimal expected welfare by $\WEL(\OPTmech,\valuecdfs)=\Ex[\vals\sim\valuecdfs]{\max_{\allocopt\in \feasible} \sum_\agent\valagent\allocoptagent}$.
Fixing a family of value distributions $\mathbb F$ (e.g.\ degenerate, product, or unrestricted), and family of equilibria $\textsc{EQ}_{\valuecdfs}(\mech)$ for each $\valuecdfs\in \mathbb F$ and mechanism $\mech$ (e.g. $\epsilon$-Nash, Bayes-Nash, Bayes coarse correlated), we study the {\em robust efficiency}
\begin{equation*}
	\min_{\valuecdfs\in\mathbb F,\jointcdf\in \textsc{EQ}_{\valuecdfs}(\mech)}\frac{\WEL(M,\jointcdf)}{\WEL(\OPTmech,\valuecdfs)}.
\end{equation*}
A high robust efficiency (close to $1$) indicates that $\mech$ is always nearly efficient, whereas a very small value suggests a pathology that might rule out $\mech$ in practice. By exposing the structural characteristics that influence worst-case performance, our framework can inform robust design. We also show that several commonly-observed formats such as the single-item first-price auction can be explained by their high robust efficiency.

For revenue, we assume independently-distributed values, and compare to the revenue-optimal Bayesian incentive compatible mechanism.
Per \citet{M81}, this optimal mechanism can depend intricately on the prior $\valuecdfs$ over values. Denote the expected revenue of the optimal mechanism $\OPTmech_\valuecdfs$ for prior $\valuecdfs$ by $\rev(\OPTmech_\valuecdfs,\valuecdfs)$. We study mechanisms which are prior-independent except for {\em monopoly reserves} (defined formally in Section~\ref{sec:resolution}), which depend on much less fine-grained information than the form of the Bayesian optimal mechanism. For these mechanisms, we study the ratio
\begin{equation*}
	\min_{\valuecdfs\in\mathbb F, \jointcdf\in \textsc{EQ}_{\valuecdfs}(\mech_\valuecdfs)}\frac{\REV(\mech_\valuecdfs,\jointcdf)}{\REV(\OPTmech_\valuecdfs,\valuecdfs)},
\end{equation*}
where $\mech_\valuecdfs$ denotes mechanism $\mech$ endowed with monopoly reserve prices for $\valuecdfs$.
\section{\RC}
\label{sec:rc}

In this section, we formally define \trc. 
\Rc\ depends only on the rules of the mechanism.
Importantly, it does not depend on preferences, information, or strategies
of agents; nor does it depend on relationships between these.  As we
will see in \cref{sec:resolution}, this parameter governs
the extent to which equilibria in the mechanism obtain good welfare
and revenue, and is robust to specifics that are typically critical for
equilibria such as preferences, information, and strategies.  

For the time being we restrict the analysis to {\em winner-pays-bid} mechanisms. 
In such a mechanism, each agent $\agent$'s action is a single real-valued bid $\bidagent$.\footnote{We use the notation $\bidagent$ for actions in single-bid mechanisms, and $\actionagent$ for actions in general mechanisms.} 
Given a single-bid allocation rule $\bidalloc$, the winner-pays-bid mechanism for $\bidalloc$ has payment rule $\bidpaymentagent(\bids)=\bidagent\bidallocagent(\bids)$.
Section~\ref{sec:beyondpyb} extends this section's analysis to general mechanisms. 
\Trc\ compares two quantities: {\em threshold surplus}, which quantifies the level of competition each agent faces, and the revenue of the mechanism. 
Hence, \trc\ measures the extent to which competition translates into revenue. 
We first consider deterministic allocation rules without reserve
prices under full information, where this comparison is particularly straightforward. 
Competition can be easily
quantified by the {\em threshold bid} that each agent's bid must exceed
to win.

\begin{definition}
  \label{def:intcumulative-det}
  In an (implicit) deterministic single-bid rule $\bidalloc$ and
  a profile of bids $\bids$, we summarize agent $\agent$'s competition by the {\em threshold bid}
  $\thresholdagent(\bidothers)=\inf\{\bidagent\,|\,\bidallocagent(\bidagent,\bidothers)
  = 1\}$ that the agent must outbid to win.  Denote the {\em
    threshold surplus} for allocation $\allocaltagent$ by
  $\cumuthreshagent{\agent}{\allocaltagent}{} =
  \thresholdagent\,\allocaltagent$. Threshold surplus for deterministic rules is depicted in Figure~\ref{fig:detrc}.
\end{definition}

\begin{definition}\label{def:detrc}
  The {\em \rc}\ of a winner-pays-bid mechanism with deterministic allocation rule $\bidalloc$ is the largest $\revpar$
  such that, for any profile of bids $\bids$ and any feasible
  allocation $\allocalt$, the revenue is at least a $\revpar$ fraction
  of the {\em threshold surplus}.
  \begin{equation}
    \label{eq:rcsimple}
    \rev(\mech,\bids)
    \geq\revpar\sum\nolimits_\agent
    \cumuthreshagent{\agent}{\allocaltagent}{}.
  \end{equation}
\end{definition}

Later in this section, we will identify a number of
deterministic mechanisms of interest for which the \rc\ defined by Definition~\ref{def:detrc} is easy to analyze.
This definition is tractable because threshold
surplus is a simple linear function of the allocation, i.e.,
$\sum\nolimits_\agent
\cumuthreshagent{\agent}{\allocaltagent}{} =\sum\nolimits_\agent
\thresholdagent\, \allocaltagent$.  
To illustrate, consider the single-item highest-bid-wins rule, which has \rc\ $\revpar = 1$.  Given bids, the threshold bids of losers
are the highest bid; the threshold bid of the winner is the second
highest bid.  Thus, both the optimal threshold surplus and the revenue are equal to the highest bid (see the formal treatment in \Cref{sec:rc-single}).

The definition of \trc\ extends naturally to allocation rules with reserves, as well as to uncertainty in the \mf. (Such uncertainty could come from randomized allocations or from random selection of participants in the auction, captured by the value distribution.) 
The general definition of \rc\ will enable quantification of welfare and revenue over the more complicated mechanisms and equilibria arising from reserves and randomness. Importantly, though, extension theorems given subsequently will allow us to characterize \trc\ in these environments via analysis of simpler mechanisms with full information and no reserves using Definition~\ref{def:detrc}.

Reserve prices make it harder for an agent to receive allocation, but
do not result in revenue when not met.
For the general definition of \rc\ to compare revenue and competition, it must therefore discount reserves.
We incorporate reserves to mechanism in the following generic way. Given profile of reserves $\reserves$ and a mechanism with single-bid rule $\bidalloc$, the new rule $\bidallocr$: (1) Solicits a bid $\bidagent$ from each agent $\agent$. (2) For each agent $\agent$, if $\bidagent < \reserveagent$, sets $\bidpostresagent = 0$, else $\bidpostresagent=\bidagent$. (3) Allocates according to $\bidalloc(\bidspostres)$. 
We typically consider adding reserves to rules where bidding $0$ guarantees an agent goes unallocated, i.e. $\bidallocagent(0,\bidothers)=0$.

With uncertainty in the \mf, the competition for an agent $\agent$ with value $\valagent$ now depends on the distribution of other agents' bids, which may be correlated with $\valagent$.
Because threshold bids are no longer deterministic, the competition faced by an agent now depends on the desired level of allocation.
For the general definition of \rc\ to quantify competition, it must account for this range of options.
To this end, define the interim allocation rule for value $\valagent$ by $\bidallocagent(\bid\,|\,\valagent)=\mathbb E_{\valothers,\bidothers\,|\,\valagent}[\bidallocagent(b,\bidothers)]$, and define its inverse as $\interthresh{\agent}{\allocation\,|\,\valagent}=\inf
\{\bid\,|\,\bidallocagent(\bid\,|\,\valagent)\geq \allocation\}$.
The inverse $\interthresh{\agent}{\allocation\,|\,\valagent}$ now quantifies agent $\agent$'s obstacles to allocation at every allocation level $\allocation$, and allows us to generalize threshold surplus.
Intuitively, an agent faces strong competition if bids above the reserve generally yield low allocation.
\begin{definition}
  \label{def:intthresh}
  \label{def:intcumulative}
  For an (implicit) single-bid allocation rule $\bidalloc$ with (implicit) reserves $\reserves$ and joint distribution $\jointcdf$ over bids and values, 
  the
  competition faced by agent $\agent$ with value $\valagent$ for obtaining allocation
  $\allocaltagent$ is summarized by the \emph{threshold surplus with
    discounted reserve}, defined as
  $\cumuthreshagent{\agent}{\allocaltagent\,|\,\valagent}{}=\int_{\bidallocagent(\reserveagent\,|\,\valagent)}^{\allocaltagent}\interthresh{\agent}{\allocation\,|\,\valagent}\,d\allocation$.
  With the
  reserve price denoted explicitly, define
  $\cumuthreshagent{\agent}{\allocaltagent\,|\,\valagent}{\reserveagent}$. 
  The general definition of threshold surplus is depicted in Figure~\ref{fig:exprc}.
\end{definition}

\begin{figure}[t]	\small
	\centering
	\subfloat[][Threshold surplus for a deterministic allocation rule with degenerate bid profile $\bids$. Agent $\agent$'s allocation steps from $0$ to $1$ at $\thresholdagent(\bidothers)$. For a desired allocation level $\allocaltagent$, the threshold surplus is $\thresholdagent(\bidothers)\allocaltagent$. \label{fig:detrc} ]{
	\begin{tikzpicture}[xscale=3.5, yscale=3.5, domain=0:0.9, smooth]
		\draw[-] (0,0) -- (0,1) node[left] {$1$};
		\draw[-][line width=1.3pt] (0,0) -- (.7,0);
		\draw[-][line width=1.3pt] (.7,0) -- (.7,.95);
		\draw[-][line width=1.3pt] (.7,.95) -- (1.3,.95);
		\draw[-] (0,0) -- (1.4,0) node[below] {Bid};
		\draw[pattern=north east lines, pattern color=lightgray] 
		( 0,0) rectangle ( 0.7, 0.5);
		\node at (-0.1,0.5) {\footnotesize $\allocaltagent$};
		\node [fill=white] at ( 0.35,0.25) {\footnotesize $ \cumuthreshagent{\agent}{\allocaltagent}{}$};
		\node at (0.7, -0.09) {\footnotesize $\thresholdagent(\bidothers)$};
		\node at (0.7,1.1) {Bid Allocation Rule};
	\end{tikzpicture}
}\IFECELSE{\hspace*{0.2cm}}{\hspace*{1cm}}
\subfloat[][Threshold surplus for a randomized allocation rule with a reserve price.
Agent $\agent$ receives no allocation for bids below $\reserveagent$. Threshold surplus is only counted for the fraction of allocations above $\bidallocation(\reserveagent)$.\label{fig:exprc} ]
{
	\begin{tikzpicture}[xscale=3.5, yscale=3.5, domain=0:0.9, smooth]
		\draw[ pattern=north west lines, pattern color=lightgray, ] (0,.1) -- (0,.7) -- (.9, .7) -- (.91,0.7) .. controls ( 0.68,0.4) and ( 0.6,0.32) .. (0.4,0.2) --(0,0.2);
		\node [fill=white] at ( 0.35,0.46) {\footnotesize $ \cumuthreshagent{\agent}{\allocaltagent}{\reserveagent}$};
		\draw[line width=1.3pt] (1.12,1) -- (.91,0.7) .. controls ( 0.68,0.4) and ( 0.6,0.32) .. (0.4,0.2) --(0.4,0);
		\draw[-] (0,0) -- (0,1) node[left] {$1$};
		\draw[-] (0,0) -- (1.4,0);
		\draw[-][line width=1.3pt] (0,0) -- (.4,0);
		\draw[-] (0,0) -- (1.4,0) node[below] {Bid};
				\node at (-0.1,0.7) {\footnotesize $\allocaltagent$};

		\node at (-0.16,0.20) {\footnotesize $\bidallocation(\reserveagent)$}; 
		\node at (0.4, -0.09) {\footnotesize $\reserveagent$}; 
		\node at (0.7,1.1) {Bid Allocation Rule};
	\end{tikzpicture}
}   
	\caption{ \label{fig:thresholds}}
\end{figure}

\Rc\ compares revenue to threshold surplus for feasible allocations.
Definition~\ref{def:detrc} defined \trc\ in terms of deterministic allocations. With randomized values, we now consider randomized, {\em interim} allocations. Taking the view of the value distribution as defining populations of agents, an interim allocation level defines a feasible allocation probability for each agent in each population. Formally, let $\allocalt$ be an ex post allocation function, mapping value profiles to feasible allocations $\allocalt(\vals)\in \feasible$. The interim allocation for agent $\agent$ with value $\valagent$ we denote by $\allocaltagent(\valagent)=\Ex[\valothers\,|\,\valagent]{\allocaltagent(\vals)\,|\,\valagent}$. Overloading notation, a profile of interim allocation functions $\allocalt=(\altallocation_1(\cdot),\ldots,\altallocation_\numagents(\cdot))$ is {\em feasible} if it is induced by some feasible ex post allocation function.

\begin{definition}\label{def:exprc}
  The {\em \rc} of a winner-pays-bid mechanism with (randomized)
  allocation rule $\bidalloc$ for value distribution $\valuecdfs$ is the largest
  $\revpar$ such that, for any joint distribution $(\valuecdfs,\bidcdfs)$ over values and bids where values follow $\valuecdfs$, and any feasible profile $\allocalt$ of interim allocation functions, the expected revenue is at least a
  $\revpar$ fraction of the {\em threshold surplus}:
  \begin{equation}
  	\label{eq:exprc}
  	\rev(\mech,\valuecdfs,\bidcdfs)
  	\geq\revpar \sum\nolimits_\agent\Ex[\valagent\sim\valuecdf]{\cumuthreshagent{\agent}{\allocaltagent(\valagent)\,|\,\valagent}{}}.
  \end{equation}
	The \rc\ for a family of allocation rules over a family of value distributions is the smallest \rc\ of any rule in the first family over any distribution in the second. We refer to a rule's \rc\ over all value distributions simply as its \rc.
\end{definition}

Definition~\ref{def:exprc} generalizes the simpler Definition~\ref{def:detrc} to randomized bid profiles over populations of participants. 
Though Definition~\ref{def:exprc} depends on the value distribution $\valuecdfs$, we show via closure properties that a rule's \rc\ on all distributions can be analyzed by considering any single distribution, including a degenerate one. 
It will therefore often suffice to work with the simpler Definition~\ref{def:detrc} to characterize the \rc\ of Definition~\ref{def:exprc}. 
Conceptually, this shows that \trc\ depends only on the rules of the mechanism.

We now derive two closure properties of \trc\ (Definition~\ref{def:exprc}). The first is that \trc\ is closed with respect to reserve
prices, i.e., if a mechanism has \rc\ $\revpar$
without reserves, then its \rc\ with
reserves is $\revpar$. Then, we consider the impact of two types of randomization on \trc. First, we consider randomization over mechanisms. In other words, if the lowest \rc\
of any rule in a family of winner-pays-bid mechanisms is $\revpar$ then the \rc\ of any convex combination of mechanisms in the family
is at least $\revpar$. Second, we study closure under mixture over priors. 
We start with closure under reserves.

\begin{lemma}
	\label{lem:dettoexp2}
	\Rc\ is closed under reserve pricing, i.e., given a mechanism $\mech$ with \rc\ $\revpar$ on $\valuecdfs$ (without reserves), \trc\ of $\mech$ with reserves on $\valuecdfs$ is also $\revpar$.
\end{lemma}

\begin{proof}[Proof of \cref{lem:dettoexp2}]
	\newcommand{\bidsres}{\bids^{\reserves}}
	\newcommand{\bidsresdist}{\actionsdist^{\reserves}}
	\newcommand{\bidsresdistothers}{\bidsresdist_{-\agent}}
	
	Let $\mech$ be a mechanism with \rc\ $\revpar$ without reserves (equivalently with reserves $\zeros$) on value distribution $\valuecdfs$.
  We will show that adding any profile of reserves $\reserves$ yields a mechanism $\mechr$ with \rc\ at least $\revpar$ with reserves $\reserves$ on $\valuecdfs$. That is, the worst-case \rc\ over  $\mech$ with any reserves is achieved by $\mech^{\zeros}$. Let $\bidalloc$ and $\bidallocr$ denote the allocation rules of $\mech$ and $\mechr$, respectively.
  
  For a bid distribution $\bidcdfs$, let $\bidcdfs^{\reserves}$ denote the bid distribution obtained by setting to $0$ all bids $\bidagent<\reserveagent$ failing to meet the reserves $\reserves$.
  The main ideas of the proof are that (a) outcomes are equivalent for
  $\bidalloc^{\zeros}$ on $\bidcdfs^{\reserves}$ and $\bidallocr$ on $\bidcdfs$, and (b) fixing the bid allocation rules, reserves only lower threshold surplus.
  Thus, a \rc\ without reserves implies the same \rc\ with reserves.
  We adopt the following notation that makes the allocation rule $\bidalloc$ and
  bid distribution $\bidcdfs$ explicit in our notation for threshold
  surplus with discounted reserve
  $\cumuthreshagent{\agent}{\allocation,\bidalloc,\bidcdfs\,|\,\valagent}{\reserveagent} = \cumuthreshagent{\agent}{\allocation\,|\,\valagent}{\reserveagent}$.
  Zero reserves will be explicitly designated as such.
  
  Consider the following analysis, with subsequent discussion:	
  \begin{align}
  	\label{eq:res-equiv1}
  	\rev(\mechr,\valuecdfs,\bidcdfs)&=\rev(\mech^{\zeros},\valuecdfs,\bidcdfs^{\reserves})
  	\\
  	\label{eq:res-revcov}
  	&\geq\revpar \sum\nolimits_{\agent}\Ex[\valagent]{\cumuthreshagent{\agent}{\allocaltagent(\valagent),\bidalloc^{\zeros},\bidcdfs^{\reserves}\,|\,\valagent}{0}}
  	\\
  	\label{eq:res-equiv2}
  	&=\revpar \sum\nolimits_{\agent}\Ex[\valagent]{\cumuthreshagent{\agent}{\allocaltagent(\valagent),\bidalloc^{(0,\reservesothers)},\bidcdfs\,|\,\valagent}{0}} 
  	\\
  	\label{eq:res-addres1}
  	&\geq \revpar\sum\nolimits_{\agent}\Ex[\valagent]{ \cumuthreshagent{\agent}{\allocaltagent(\valagent),\bidalloc^{(0,\reservesothers)},\bidcdfs\,|\,\valagent}{\reserveagent}}
  	\\
  	\label{eq:res-addres2}
  	&= \revpar\sum\nolimits_{\agent}\Ex[\valagent]{ \cumuthreshagent{\agent}{\allocaltagent(\valagent),\bidallocr,\bidcdfs\,|\,\valagent}{\reserveagent}}.
  \end{align}

Equations~\eqref{eq:res-equiv1} and~\eqref{eq:res-equiv2} follow by
the equivalence of outcomes from reserves in the allocation rule and
reserves in the bids. Equation~\eqref{eq:res-revcov}
follows from the assumed \rc\ of $\mech^{\zeros}$. Equation~\eqref{eq:res-addres1} follows
from the definition of threshold surplus with discounted reserves;
i.e., discounting reserves lowers the threshold surplus.
Equation~\eqref{eq:res-addres2} follows because $\cumuthreshagent{\agent}{\cdot\,|\,\valagent}{\reserveagent}$ considers only allocation levels above $\bidallocagent(\reserveagent\,|\,\valagent)$. Combining the sequence of inequalities we observe that $\mechr$ has
\rc\ at least $\revpar$ on bid distribution $\bidcdfs$.
\end{proof}

Closure under convex combination of both mechanisms and value distributions will follow from a single convexity argument. The definitions of each are below, followed by a proof encompassing both.

\begin{definition}
	Let $\convcomb\sim U[0,1]$ be a uniform random variable indexing over mechanisms $\mech^{\convcomb}$ with allocation rules $\bidalloc^\convcomb$ and feasibility environments $\feasible^\convcomb$.  The
	convex combination of these mechanisms has allocation rule $\bidalloc(\bids) =
	\expect[\convcomb]{\bidalloc^\convcomb(\bids)}$.  The corresponding feasibility environment is $\feasible = \{ \alloc =
	\expect[\convcomb]{\alloc^\convcomb} : \forall \convcomb \in[0,1],\ \alloc^\convcomb \in
	\feasible^\convcomb \}$.
\end{definition}

\begin{definition}
	Let $\convcombdist\sim U[0,1]$ be a uniform random variable indexing over value distributions $\valuecdfs^{\convcombdist}$. To draw from the convex combination of these distributions $\valuecdfs$, draw $\convcombdist\sim U[0,1]$ then draw $\vals\sim\valuecdfs^{\convcombdist}$.
\end{definition}

\begin{lemma}
	\label{lem:convandnorm}
	\label{lem:convex-combination}
	\Rc\ is closed under convex combination of allocation rules and of value distributions, i.e.\ (i) if a rule $\bidalloc$ has \rc\ $\revpar$ on a family of priors $\mathscr F$, then it has \rc\ $\revpar$ on the family of convex combinations over $\mathscr F$; (ii) if the family of rules $\mathscr X$ has \rc\ $\revpar$ on prior $\valuecdfs$, then the family of convex combinations over $\mathscr X$ also has \rc\ $\revpar$ on $\valuecdfs$.
\end{lemma}

\begin{proof}
	Let $\mech$ be a convex combination of mechanisms indexed by $\convcomb$ (with corresponding allocation rules $\bidalloc^{\convcomb}$, feasibility settings $\feasible^\convcomb$, and combinations $\bidalloc$ and $\feasible$). Further let $\valuecdfs$ be a convex combination of distributions, indexed by $\convcombdist$. Assume for all $\convcomb$ and $\convcombdist$, $\mech^\convcomb$ has \rc\ at least $\revpar$ on $\valuecdfs^{\convcombdist}$. We will argue that $\mech$ has \rc\ at least $\revpar$ on $\valuecdfs$. This implies both stated claims, as either convex combination could be trivial.

Let $\allocalt$ map value profiles to allocations feasible for the convex combination environment $\feasible$.
Then for each profile $\vals$, we can write $\allocalt(\vals) = \Ex[\convcomb]{\allocalt^{\convcomb}(\vals)}$ for some collection of allocation functions $\allocalt^{\convcomb}$ respectively feasible for $\feasible^{\convcomb}$. 
Moreover, for an agent $\agent$ with value $\valagent$ the interim allocation $\allocaltagent(\valagent)=\Ex[\vals]{\allocaltagent(\vals)\,|\,\valagent}$ satisfies $\allocaltagent(\valagent)=\expect[\convcomb,\convcombdist]{\allocaltagent^{\convcomb,\convcombdist}(\valagent)\,|\,\valagent}$, where $\allocaltagent^{\convcomb,\convcombdist}(\valagent)$ is the interim allocation with respect to $\allocalt^{\convcomb}$ under distribution $\valuecdfs^{\convcombdist}$. 
Similarly, let $\bidcdfs$ be a distribution of bids (that may be correlated with $\valuecdfs$).
We may draw from the joint distribution of bids and values by first drawing $\convcombdist$, then drawing $\vals$ and $\bids$ jointly from $\valuecdfs^{\convcombdist}$ and $\bidcdfs^{\convcombdist}$ for suitably defined bid distributions $\bidcdfs^{\convcombdist}$.

Now consider an agent $\agent$ with value $\valagent$, and consider an allocation level $\allocaltagent(\valagent) = \Ex[\convcomb,\convcombdist]{\allocaltagent^{\convcomb,\convcombdist}(\valagent)\,|\,\valagent}$.
Let $\cumuthreshagent{\agent}{\allocaltagent^{\convcomb,\convcombdist}(\valagent)\,|\,\valagent}{\convcomb,\convcombdist}$ denote $\agent$'s threshold surplus with respect to $\mech^\convcomb$, $\bidcdfs^{\convcombdist}$, and $\valuecdfs^{\convcombdist}$, and $\cumuthreshagent{\agent}{\allocaltagent(\valagent)\,|\,\valagent}{}$ the threshold surplus with respect to $\mech$, $\bidcdfs$, and $\valuecdfs$.
 The following inequalities show that $\expect[\convcomb,\convcombdist]{\cumuthreshagent{\agent}{\allocaltagent^{\convcomb,\convcombdist}(\valagent)\,|\,\valagent}{\convcomb,\convcombdist}\,|\,\valagent}\geq \cumuthreshagent{\agent}{\allocaltagent(\valagent)\,|\,\valagent}{}$, explained after their statement:
\begin{align*}
	\cumuthreshagent{\agent}{\allocaltagent(\valagent)\,|\,\valagent}{}
	&= \int_0^{\infty} \max(\allocaltagent(\valagent) - \bidallocagent(z\,|\,\valagent),0)\,dz\\
	&= \int_0^{\infty} \max\Big(\expect[\convcomb,\convcombdist]{\allocaltagent^{\convcomb,\convcombdist}(\valagent) - \bidallocagent^{\convcomb,\convcombdist}(z\,|\,\valagent)\,\big|\,\valagent},0\Big)\,dz\\
	& \leq \expect[\convcomb,\convcombdist]{\int_0^{\infty} \max(\allocaltagent^{\convcomb,\convcombdist}(\valagent)-\bidallocagent^{\convcomb,\convcombdist}(z\,|\,\valagent),0)\,dz\,\Big|\,\valagent}\\
	& = \expect[\convcomb,\convcombdist]{\cumuthreshagent{\agent}{\allocaltagent^{\convcomb,\convcombdist}(\valagent)\,|\,\valagent}{\convcomb,\convcombdist}\,\big|\,\valagent}.
\end{align*}
The first and last equalities are another way to write the integral defining the threshold surplus.  The second equality is from the definitions of convex combination, and linearity of expectation.  The inequality follows from convexity of the function $\max(\cdot,0)$ and linearity of integration.

Now we write out the definition of the \rc\ for feasible allocation $\allocalt$:
\begin{align*}
	\rev(\mech,\valuecdfs,\bidcdfs)
	&=\Ex[\convcomb,\convcombdist]{\rev(\mech^\convcomb,\valuecdfs^{\convcombdist},\bidcdfs^{\convcombdist})}\\
	&\geq \mathbb
	 E_{\convcomb,\convcombdist}\left[\revpar\sum\nolimits_{\agent}\Ex[\valagent\sim\valuecdf^{\convcombdist}]{\cumuthreshagent{\agent}{\allocaltagent^{\convcomb,\convcombdist}(\valagent)\,|\,\valagent}{\convcomb,\convcombdist}}\right]\\
	 &= \revpar
	 \sum\nolimits_{\agent}\Ex[\valagent\sim\valuecdfs]{\mathbb E_{\convcomb,\convcombdist\,|\,\valagent}\left[\cumuthreshagent{\agent}{\allocaltagent^{\convcomb,\convcombdist}(\valagent)\,|\,\valagent}{\convcomb,\convcombdist}\right]}\\
	 &\geq\revpar
	 \sum\nolimits_{\agent}\Ex[\valagent\sim\valuecdfs]{\cumuthreshagent{\agent}{\allocaltagent(\valagent)\,|\,\valagent}{}}.
\end{align*}

The first line follows from the definitions of convex combination and revenue. The second follows from applying \trc. The third line follows from properties of expectation, and the final line follows because $\expect[\convcomb,\convcombdist]{\cumuthreshagent{\agent}{\allocaltagent^{\convcomb,\convcombdist}(\valagent)\,|\,\valagent}{\convcomb,\convcombdist}\,|\,\valagent}\geq \cumuthreshagent{\agent}{\allocaltagent(\valagent)\,|\,\valagent}{}$, as argued above. We conclude that $\mech$ has \rc\ at least $\revpar$ on $\valuecdfs$. 
\end{proof}

Definition~\ref{def:exprc} formulates \rc\ with an explicit dependence on the value distribution. To analyze mechanisms robustly, we seek to understand \rc\ in the worst case across distributions. The following discussion shows that as a consequence of closure under convex combination, it will suffice to analyze only degenerate distributions.

First, notice that under a degenerate value distribution, the definition of \rc\ simplifies, and is equivalent to the statement that for any bid distribution $\bidcdfs$ and fixed allocation profile $\allocalt\in\feasible$,
\begin{equation*}
	\rev(\mech,\bidcdfs)\geq\revpar \sum\nolimits_\agent\cumuthreshagent{\agent}{\allocaltagent}{},
\end{equation*}
where degeneracy of the value distribution allows us to omit conditioning from $\cumuthreshagent{\agent}{\allocaltagent}{}$. This inequality no longer depends on the value distribution, and when $\bidcdfs$ is also degenerate (i.e.\ a fixed bid profile), matches Definition~\ref{def:detrc} for deterministic allocation rules. Since the \rc\ for degenerate value distributions does not depend on the value profile chosen, we obtain, it must be the same for all choices of value profile. Hence:
\begin{lemma}\label{lem:degen}
	The \rc\ of any winner-pays-bid mechanism is the same on all degenerate distributions.
\end{lemma}

Any distribution can be written as a convex combination over degenerate ones. As a consequence of Lemma~\ref{lem:convex-combination}, we therefore obtain:

\begin{corollary}\label{cor:degen}
	A winner-pays-bid mechanism's \rc\ is equal to its \rc\ on only degenerate distributions.
\end{corollary}

Lemma~\ref{lem:degen} and Corollary~\ref{cor:degen} together imply that characterizing a mechanism's \rc\ on any fixed, degenerate value distribution will also tightly characterize its worst-case \rc\ over all distributions.
As we analyze the \rc\ of specific mechanisms in Sections~\ref{sec:rc-single}-\ref{sec:rc-gfp}, this tightness will greatly simplify the analysis.

 In
Section~\ref{sec:rc-single}, we consider the highest-bid-wins mechanism for single-item multi-unit auctions, and prove a \rc\ of $1$. For contrast, we then turn the multi-item
setting of single-minded combinatorial auctions in
Section~\ref{sec:rc-hbw}. We show that for multiple items,
highest-bids-win, which is welfare-optimal in the absence of incentives, has an
undesirable \rc. In Section~\ref{sec:rc-greedy}, we
consider the \rc\ of greedy
mechanisms. We observe that they generally lack the pathology of highest-bids-win single-minded combinatorial auctions. Finally, we demonstrate the usefulness of closure under convex combination by
considering position auctions in Section~\ref{sec:rc-gfp}. 

In
Section~\ref{sec:resolution}, we show how these \rc\ bounds imply a variety of robust welfare and revenue guarantees under several standard notions of equilibrium. We define individual efficiency to measure the impact of the equilibrium concept on robust welfare. Together, \rc\ and individual efficiency govern a mechanism's robust performance.

\subsection{Single-item Multi-unit Auctions}
\label{sec:rc-single}

This section considers the $n$-agent single-item and multi-unit highest-bids-win mechanisms for unit-demand agents.  Under the $k$-unit highest-bids-win rule, each agent $\agent$
submits a bid $\bidagent$, the and the $k$ highest bidders win a unit. This generalizes the standard allocation rule for single-item auctions. By Lemma~\ref{cor:degen}, it suffices to analyze degenerate priors. Since the highest-bids-win rule and $k$-unit environment are both deterministic, we may consider the simpler Definition~\ref{def:detrc}. We first give the proof for single-item environments.

\begin{theorem}
  \label{thm:first-price-rc}
  The highest-bids-win winner-pays-bid mechanism has \rc\ $1$ in single-item environments.
\end{theorem}

\begin{proof}
  The revenue under bid profile $\bids$ is
  the highest bid, i.e.,  $\sum_\agent \bidagent\bidallocagent(\bids) = \max_\agent \bidagent$.  Each agent's threshold bid
  $\thresholdagent(\bidothers)$ is at most the highest bid.  Thus, for any feasible allocation $\allocalt$, i.e., with $\sum_\agent \allocaltagent \leq 1$, we can bound the threshold surplus by the bid surplus. 
  \begin{align*}
    \sum\nolimits_\agent \thresholdagent(\bidothers)\,\allocaltagent
    &\leq \max\nolimits_\agent \bidagent \sum\nolimits_\agent \allocaltagent\\
    &\leq \sum\nolimits_\agent \bidagent\bidallocagent(\bids).
  \end{align*}
In particular, choosing $\allocaltagent=1$ for an $\agent$ other than the highest bidder causes the above inequalities to hold with equality.
\end{proof}

A very similar proof shows that the highest-bids-win rule for selling $k$-units to unit-demand agents has \rc\ $1$.  Rather than give the elementary proof here, we will
observe it as a corollary of \cref{thm:matroid-rc}, given subsequently.

\begin{theorem}
	\label{thm:multi-unit-rc}
	For any multi-unit environment, the
        highest-bids-win winner-pays-bid mechanism has \rc\ 1.
\end{theorem}

\subsection{Single-Minded Combinatorial Auctions: Highest Bids Win}
\label{sec:rc-hbw}

In Section~\ref{sec:rc-single}, we saw that in single-item and $k$-unit settings, the mechanism that allocates the highest bidders has \rc\ $1$. We now present a negative example, and show that in another natural multi-item setting, a generalization of this same rule has poor competitive efficiency. That is, competition corresponds less directly to revenue. This multi-item setting and allocation mechanism are as follows:

\begin{definition}
	A \emph{single-minded combinatorial auction} feasibility environment
	is defined by $\numitems$ indivisible items, $n$ agents that each desire
	a bundle of items, and the constraint that no item can be allocated
	more than once.  Agent $\agent$ desires the set of items
	$\wantsagent$, she receives value $\valagent$ for receiving any
	superset of $\wantsagent$ and value $0$ otherwise.  An allocation vector
	$\alloc\in\{0,1\}^n$ is feasible if and only if for all agents $\agent \neq \agentalt$, simultaneous allocation $\allocagent = \allocagent[\agentalt] = 1$ implies disjoint demands $\wantsagent \cap \wantsagent[\agentalt] = \emptyset$.
\end{definition}

\begin{definition}
	The \emph{highest-bids-win winner-pays-bid mechanism} allocates the feasible set of bidders with the highest total bid, and charges each winner their bid.
\end{definition}

In a single-minded combinatorial auction under the highest-bids-win mechanism, serving a single bidder may simultaneously block many others. As an example, consider an $\numitems+1$-agent environment, where $\wantsagent[\agent]=\{\agent\}$ for $\agent\in\{1,\ldots,\numitems\}$ and $\wantsagent[\numitems+1]=\{1,\ldots,\numitems\}$.
 Agent $\numitems+1$ is mutually exclusive with any subset of the other agents. The impact of agent $\numitems+1$ on the competition experienced by agents $1,\ldots,\numitems$ far exceeds their impact on revenue. Under the bid profile $\bids=(0,\ldots,0,1)$, the revenue is $1$, while each agent $\agent\in\{1,\ldots,\numitems\}$ faces a threshold bid of $\thresholdagent(\bidothers)=1$. Since the allocation vector $\allocalt=(1,\ldots,1,0)$ is feasible, this immediately implies:

\begin{lemma}
	\label{lem:combbad}
	There exists a single-minded combinatorial auction environment where the highest-bids-win winner-pays-bid mechanism has competitive efficiency at most $1/\numitems$.
\end{lemma}


\subsection{Greedy Auctions}
\label{sec:rc-greedy}

The previous section demonstrates the inability of the highest-bids-win mechanism to effectively convert competition into revenue in combinatorial multi-item settings. We now show that other mechanisms may manage this relationship more effectively. In particular, this section considers mechanisms with {\em greedy} allocation rules, defined below.

\begin{definition}
	The \emph{greedy by priority} rule is given by a profile $\boldsymbol\psi=(\psi_1,\ldots,\psi_\numagents)$ of nondecreasing priority functions mapping bids for each agent $\agent$ to real numbers. It proceeds in the following way:
	\begin{enumerate}
		\item Sort agents in nonincreasing order of priority $\psi_\agent(\bidagent)$.
		\item Initialize the set of winners $S=\emptyset$.
		\item For each agent $\agent$ in sorted order: if $S\cup\{\agent\}$ is feasible, set $S=S\cup\{\agent\}$.
		\item Return $S$.
	\end{enumerate}
\end{definition}

For example, the greedy by bid rule is given by priority functions $\psi_\agent(\bidagent)=\bidagent$ for all $\agent$. Greedy by priority rules may be defined in any feasibility environment. In many settings, including the single-minded combinatorial auction, greedy rules will be suboptimal in the absence of incentives --- they may not select a set of winners with highest total bid. We will show that this suboptimality completely governs the \rc\ of greedy auctions. Since the design of approximately optimal greedy algorithms is well-studied, we obtain several \rc\ bounds as immediate corollaries. We first define the measure of approximate optimality.

\begin{definition}
	A bid allocation rule $\bidalloc$ is an {\em$\alpha$-approximation} for a feasibility environment $\feasible$ if for any bid profile $\bids$ and feasible allocation $\allocalt$, we have:
	\begin{equation}
	\sum\nolimits_\agent \bidagent\bidallocagent(\bids)\geq\alpha\sum\nolimits_\agent\bidagent\allocaltagent.
	\notag
	\end{equation}
\end{definition}

We formalize the relationship between approximation and the \rc\ as follows.

\begin{theorem}
	\label{thm:greedy}
	For any feasibility environment $\feasible$, if $\bidalloc$ is a $\alpha$-approximation $\bidalloc$ for $\feasible$, then the winner-pays-bid mechanism for $\bidalloc$ has \rc\ at least $\alpha$.
\end{theorem}

For single-minded combinatorial auctions, a greedy $1/\sqrt{\numitems}$-approximation is well-known.

\begin{lemma}[\citealp{LOS02}]
	\label{lem:greedyapx}
	For any single-minded combinatorial auction environment, greedy by priority with $\psi_\agent(\bidagent)=\bidagent/\sqrt{|\wantsagent|}$ for all agents $\agent$ is a $1/\sqrt{\numitems}$-approximation.
\end{lemma}

Another family of settings where greedy allocation rules are of particular interest are {\em matroids}, defined below, where the greedy by bid rule is known to be optimal i.e.\ a $1$-approximation, absent incentives. Notable examples of matroids include $k$-unit environments, discussed in Section~\ref{sec:rc-single}, and \emph{transversal matroids}, which are matchable subsets of vertices on one side of a bipartite graph.

\begin{definition}
	A feasibility environment $\feasible$ is a \emph{matroid} if the following two properties hold:
	\begin{itemize}
		\item[i.] (Downward Closure) For any $S\in \feasible$ and $i\in S$, $S\setminus\{i\}\in \feasible$.
		\item[ii.] (Augmentation Property) For any $S_1,S_2\in \feasible$ with $|S_1|>
		|S_2|$, there exists $i\in S_1\setminus S_2$ such that $S_2\cup\{i\}\in \feasible$.
	\end{itemize}
\end{definition}

\begin{lemma}
  \label{lem:matroid}
	For any matroid feasibility environment, greedy by priority with $\psi_\agent(\bidagent)=\bidagent$ for all agents $\agent$ is a $1$-approximation.
\end{lemma}

Combining Theorem~\ref{thm:greedy} with Lemmas~\ref{lem:greedyapx} and \ref{lem:matroid} yields the following:

\begin{theorem}
	\label{thm:single-minded-rc}
	For any single-minded combinatorial auction environment, the winner-pays-bid mechanism for the
	greedy-by-priority rule with priority
	function
	$\psi_\agent(\bidagent)=\bidagent/\sqrt{|\wantsagent|}$ for
	agent $\agent$, has \rc\ at least $1/\sqrt{\numitems}$.
\end{theorem}

\begin{theorem}
	\label{thm:matroid-rc}
	For any matroid environment, the highest-bids-win winner-pays-bid
	mechanism  has \rc\ 1.
\end{theorem}

To prove Theorem~\ref{thm:greedy}, it is helpful to compare the behavior of greedy mechanisms to the non-greedy highest bids win mechanism of \cref{sec:rc-hbw}, which had a poor \rc. In the example which proved \cref{lem:combbad}, the high bid of the $(\numitems+1)$th agent discouraged participation from the others --- individually, each agent would have needed to bid
$1+\epsilon$ to win. As a group, though, the losing agents could have
won by increasing each of their bids by a tiny amount. Greedy
rules lack this pathology. For any greedy
allocation rule, we could increase the bids of every losing agent
to their threshold without changing the outcome. We formalize this
property as follows.

\begin{definition}
	Bid allocation rule $\bidalloc$ is \emph{coalitionally non-bossy} if: for
	any profiles of bids $\bids$ and $\bidsalt$ where the bids in
	$\bidsalt$ are the same as $\bids$ for winners under $\bids$
	and at most their critical prices for losers under $\bids$, i.e.,
	if $\bidallocagent(\bids)=0$ then $\bidaltagent \leq
	\thresholdagent(\bidothers)$; then the allocations of $\bidalloc$ under $\bids$
	and $\bidsalt$ are the same, i.e.
	$\bidalloc(\bids)=\bidalloc(\bidsalt)$.
\end{definition}

\begin{lemma}\label{lem:bossy}
	Any greedy by priority allocation rule is coalitionally non-bossy.
\end{lemma}

\begin{proof}
	Imagine changing $\bids$ to $\bidsalt$ by increasing one loser's
	bid at a time. Each time we increase a bid, say, of agent $\agent$,
	two things remain true: (1) $\agent$ still loses: as long as
	$\bidaltagent\leq\thresholdagent(\bidothers)$, $\agent$ is
	passed over as infeasible when she is reached by the greedy
	rule; and (2) the threshold of every other losing agent $\agentalt$
	remains unchanged: each losing agent's
	threshold is only determined by the bids of the agents who win.
\end{proof}

\begin{lemma}\label{lem:bossapprox}
	Any winner-pays-bid mechanism with a coalitionally non-bossy allocation rule has \rc\ at least its approximation ratio.
\end{lemma}

\begin{proof}
		Let $\allocalt$ be a feasible allocation, $\bids$ a profile of bids, and let $\bidsalt$ be a vector of bids where losers under $\bids$ bid $\threshpt{\bidothers}$, while winners bid as before. The following inequalities hold, with justifications after.
	\begin{align*}
		\sum\nolimits_\agent \bidagent\bidallocagent(\bids) &= \sum\nolimits_\agent \bidaltagent\bidallocagent(\bids)&\\
		&= \sum\nolimits_\agent \bidaltagent\bidallocagent(\bidsalt)\\
		&\geq\alpha\sum\nolimits_\agent\bidaltagent\allocaltagent\\
		&\geq\alpha\sum\nolimits_\agent \threshpt{\bidothers}\allocaltagent.
	\end{align*}
	The first line holds because $\bidsalt$ differs from $\bids$ only on the bids of losing agents. The second follows from the coalitional non-bossiness of greedy rules, and the third from the assumption that the greedy rule is an $\alpha$-approximation. The last line follows from the fact that $\bidsalt$ doesn't change the bids of winners under $\bids$, and for those agents, $\bidagent\geq\threshpt{\bidothers}$.
\end{proof}

Theorem~\ref{thm:greedy} follows from combining Lemma~\ref{lem:bossy} with Lemma~\ref{lem:bossapprox}.

\subsection{Position Auctions}
\label{sec:rc-gfp}
	\newcommand{\permdist}{\Pi}
\newcommand{\allocaltup}{\allocalt'}
\newcommand{\allocaltagentup}{\allocaltagent'}
\newcommand{\posthreshagent}[1]{\threshold_{\agent}^{#1}}

The mechanisms and environments in
\cref{sec:rc-single}-\cref{sec:rc-greedy} have been deterministic. We
now consider the canonical and inherently randomized allocation
environment of position auctions and show that the most natural
single-bid mechanism
has \rc\ $1$.  This result follows directly
from the \rc\ of multi-unit highest-bids-win mechanism
(\cref{thm:multi-unit-rc}) and closure of the \rc\ under
convex combination (\cref{lem:convex-combination}).

Position environments are a standard model for internet advertising
auctions, e.g., \citet{var-07} and \citet{EOS07}. Advertisers (agents)
compete for ad placement in positions in a list on a webpage. Each
position has an associated clickthrough probability, and an agent is
considered allocated when clicked. Feasible allocations are
assignments of ads to positions, with the possibility for the auctioneer
to exclude any agent from the auction. 

\begin{definition}
A {\em position environment} is given by position allocation probabilities
$1\geq\slotprob_1\geq \ldots\geq \slotprob_\numagents\geq 0$. An allocation
vector $\bidalloc$ is feasible if there exists a permutation $\perm$
over $\{1,\ldots, \numagents\}$ such that for all agents $\agent$
$\bidallocagent\in\{0,\slotprob_{\perm(\agent)}\}$, or if it is a convex combination of such vectors.
\end{definition}

The natural extension of highest-bids-win to position environments is the following:

\begin{definition}
	The {\em rank-by-bid} rule for position auctions assigns agents to positions in order of their bid. The
	agent in position $j$ wins with probability $\slotprob_j$.
\end{definition}


The following interpretation of position auctions as a convex
combination of multi-unit auctions is well known in the literature,
e.g., \citet{DHY15}.  The subsequent theorem combines this
interpretation with the closure under convex combination  of mechanisms (\cref{lem:convex-combination}) and the \rc\ of multi-unit highest-bids-win auctions (\cref{thm:multi-unit-rc}).

\begin{lemma}
  \label{lem:position-auctions=convcomb-multi-unit}
  The allocation rule of the generalized first price auction for
  position weights $\slotprob_1\geq \ldots\geq \slotprob_\numagents$
  is equivalent the convex combination of single-item multi-unit
  highest-bids-win allocation rules where $k$-units are sold with
  probability $\slotprob_k - \slotprob_{k+1}$ (with $\slotprob_{n+1} = 0$).
\end{lemma}

\begin{theorem}
  \label{thm:gfp-rc}
  For any position environment, the winner-pays-bid mechanism for the rank-by-bid rule
  has \rc\ $1$.
\end{theorem}

\section{Welfare and Revenue Analysis}
\label{sec:resolution}
In the previous section, we studied the way mechanisms manage inter-agent competition. We focused on the mechanisms' rules in isolation, independent of the behavioral assumptions on participating agents and information structure of the game. This section turns the focus to agents' beliefs and behavior. An equilibrium presents agents with a range of attainable allocation levels at different prices. Behavioral assumptions such as best response dictate the choice from this range, which in turn governs the agent's contribution to social welfare. In what follows, we define {\em\vc}, which quantifies how an agent's solution to their bidding problem impacts welfare. \Vc\ will depend on the shape of the allocation rules an agent faces and the way the agent selects their bid (e.g.\ the degree of approximate best response).
We then show how \vc\ and \trc\ combine to give robust efficiency guarantees.
Throughout, we continue to focus on winner-pays-bid mechanisms.

Equilibrium induces for an agent with value $\val$, a single-agent interim mechanism.
This mechanism is summarized by its bid allocation rule.
Depending on the information the agent possesses, this interim allocation rule may change, as may the agent's bid in response.

\begin{definition}\label{def:bo}
	A {\em single-agent \bo}\ is given by a value $\val\geq 0$ and a joint distribution $\mechdist=(\bid^\info,\bidallocation^\info)$ over pairs of bid allocation rules and bids, indexed uniformly by by $\info$.
	The agent's utility is given by $\util(\mechdist)= \expect[\info]{(\val-\bid^\info)\bidallocation^\info(\bid^\info)}$.
\end{definition}

An agent's competition is summarized by the expected strength of the bid allocation rules they face. As in Section~\ref{sec:rc} we then measure this strength by the threshold surplus $\cumuthreshagent{}{\cdot}{}$.

\begin{definition}
	In a single-agent \bo, the agent's {\em expected allocation rule} at bid $\bid$ is $\bidallocation(\bid)=\expect[\info]{\bidallocation^\info(\bid)}$, with inverse $\interthresh{}{x}=\inf\{b\,|\,\bidallocation(\bid)\geq x\}$. The {\em threshold surplus} is $\cumuthresh(z)=\int_0^z\interthresh{}{x}\,dx$.
\end{definition}

 Individual efficiency measures the tradeoff between this competition and utility: either the agent's utility is a large fraction of their value, or their competition is high.
\begin{definition}\label{def:vc}
	The {\em \vc}\ of a single-agent \bo\ $(\val,\mechdist)$ is the largest $\res$ such that for all target allocation levels $z\in[0,1]$, $\util(\mechdist) + \cumuthresh(z)\geq\res\val z$. The \vc\ of a family of single-agent \bo\ is the smallest \vc\ of any \outcome\ in the family.
\end{definition}
Definition~\ref{def:vc} can be motivated by considering a single-buyer single-seller setting, where \vc\ directly implies welfare bounds. Specifically, consider a seller with an outside option for the item, defined as follows:
\begin{definition}
	In the {\em single-buyer winner-pays-bid instance} for bidding problem $(\val,\mechdist)$, there is single seller offering an item to a single buyer.
	The seller has an outside option value $\val_0$ for the item, obtained whenever the item goes unsold, with $\val_0$ drawn by first drawing $\info\sim U[0,1]$, then drawing $\val_0$ according to CDF $\bidallocation^\info$.
	The  buyer submits bid $\bid^\info$ for each $\info$, wins the item whenever $\bid^\info\geq v_0$, and pays their bid when they win.
\end{definition}
In this setting, the bidder's competition is an explicitly-modeled outside option for the seller.
Note that the buyer's bid $\bid^\info$ is correlated with the seller's outside option value $\val_0$, with the correlation indexed as in the single-agent bidding problem by variable $\info$.
Conditioned on $\info$, the buyer's bid allocation rule is exactly $\bidallocation^\info$, induced by the rule that they win whenever $\bid^\info\geq \val_0$. Proposition~\ref{prop:single} demonstrates that individual efficiency directly implies a welfare guarantee.
\begin{proposition}\label{prop:single}
	Let $(\val,\mechdist)$ be a single-agent bidding problem with \vc\ $\eta$. Then in the single-buyer winner-pays-bid mechanism for $(\val,\mechdist)$, the social welfare is at least a $\eta$-fraction of optimal.
\end{proposition}
\begin{proof}
	Given value $\val$, the optimal social welfare is $\val\prob[]{\val_0\leq \val}+\expect[]{\val_0\,|\,\val_0\geq \val}\prob[]{\val_0\geq \val}$: the item is assigned to the seller or the buyer depending on whether $\val$ or $\val_0$ is higher. 
	Letting $z^*=\prob[]{\val_0\leq \val}$, we may write the optimal welfare equivalently as $\val z^* + \expect[]{v_0\,|\,v_0\geq v}\prob[]{v_0\geq v}$.
	
	Now consider the outcome when the buyer plays $\bid^\info$. The social welfare is the sum of three terms: the buyer's utility, the seller's revenue, and the seller's utility from the outside option.
	The buyer's utility is $\expect[\info]{(\val-\bid^\info)\bidallocation^\info(\bid^\info)}=\util(\mechdist)$.
	The seller's revenue is $\bid^\info\prob[]{\bid^\info\geq \bid^\info}$.
	The seller's utility from the outside option is $\expect[]{\val_0\,|\,\bid^\info<\val_0}\prob[]{\bid^\info< \val_0}$.
	We can lower bound the seller's total utility in the following way:
	\begin{align*}
		&\mathbb E[\bid^\info\,|\,\bid^\info\geq \val_0]\text{Pr}[\bid^\info\geq \val_0]+\mathbb E[\val_0\,|\,\bid^\info< \val_0]\text{Pr}[\bid^\info< \val_0]\\
		&\quad\quad\geq \mathbb E[\val_0\,|\,\bid^\info\geq \val_0]\text{Pr}[\bid^\info\geq \val_0]+\mathbb E[\val_0\,|\,\bid^\info< \val_0]\text{Pr}[\bid^\info< \val_0]\\
		&\quad\quad=\mathbb E[\val_0\,|\,\val\geq \val_0]\text{Pr}[\val\geq \val_0]+ \mathbb E[\val_0\,|\,\val< \val_0]\text{Pr}[\val< \val_0]\\
		&\quad\quad= \cumuthresh(z^*)+\mathbb E[\val_0\,|\,\val< \val_0]\text{Pr}[\val< \val_0],
	\end{align*}
	Where the equality of the second and third lines holds because both lines are a rewriting of $\mathbb E[\val_0]$ and the fourth from the definitions of $\cumuthresh$ and $z^*$.
	We can combine all three components of the single-buyer mechanism's welfare to obtain a lower bound in terms of the optimal welfare:
	\begin{align*}
	&\mathbb E_\info[(\val-\bid^\info)\bidallocation^\info(\bid^\info)]+\mathbb E[\bid^\info\,|\,\bid^\info\geq \val_0]\text{Pr}[\bid^\info\geq \val_0]+\mathbb E[\val_0\,|\,\bid^\info< \val_0]\text{Pr}[\bid^\info< \val_0]\\
	&\quad\quad\geq\util(\mechdist)+\cumuthresh(z^*)+\mathbb E[\val_0\,|\,\val< \val_0]\text{Pr}[\val< \val_0]\\
	&\quad\quad\geq \res \val z^* + \expect[]{v_0\,|\,v_0\geq v}\prob[]{v_0\geq v}\\
	&\quad\quad\geq\res(\val z^* + \expect[]{v_0\,|\,v_0\geq v}\prob[]{v_0\geq v}),
	\end{align*}
	The second inequality follows from individual efficiency, and the third because $\res\leq 1$.
%
%
\end{proof}
In mechanisms with multiple agents competing for service, the seller's outside option is endogenously generated: rather than selling to agent $\agent$, they consider some subset of the others. We will see that \rc\ measures the value of the seller's endogeously generated outside option, and hence \rc\ and \vc\ will combine to imply robust efficiency.
To enable robust analysis, Section~\ref{sec:vcanalyses} will consider \vc\ taken in the worst case over broad families of single-agent \bos\ likely to arise in equilibrium.
We then connect the single-agent analyses to performance guarantees in auction equilibria in Section~\ref{sec:perf}. If equilibrium induces single-agent mechanisms with \vc\ $\res$ in a mechanism with \rc\ $\revpar$, then we show in Section~\ref{sec:welf} that the mechanism has welfare approximation $\revpar\res$. Furthermore, we show in Section~\ref{sec:welfare-bnerev} for revenue under Bayes coarse correlated equilibrium and independently distributed values, \vc\ $\res$ and \rc\ $\revpar$ imply a revenue approximation ratio of $\revpar\res/2$ for winner-pays-bid mechanisms with carefully selected reserve prices. Finally, we trace out the limits of this approach by exhibiting examples with low welfare in Section~\ref{sec:lb}.

%
\subsection{\VC\ Analyses}
\label{sec:vcanalyses}
In this section, we consider single-agent \bos\ that arise in equilibria of multi-agent mechanisms. In Section~\ref{sec:welfare-epsne}, we study \bos\ with a single, deterministic $\bidallocation^\info$ (i.e.\ a fixed threshold) under best response bids, as arises in pure Nash equilibria of deterministic mechanisms.
Section~\ref{sec:bnewelfare} then considers general distributions over allocation rules under best response, which will capture equilibria of randomized environments with incomplete information.
Finally, we show in Section~\ref{sec:approx} that \vc\ degrades smoothly with the agent's level of best response, enabling the study of approximate equilibria.

\subsubsection{Deterministic Rules Under Best Response}
\label{sec:welfare-epsne}

Pure strategy Nash equilibria of deterministic mechanisms induce particularly simple bidding problems, where each agent's allocation steps up to a threshold. The corresponding family of single-agent \bos\ is the following.

\begin{definition}\label{def:bodet}
	A single-agent \bo\ $(\val,\mechdist)$ is {\em deterministic} if (1) the joint distribution $\mechdist$ over bids and allocation rules $(\bid^\info,\bidallocation^\info)$ is degenerate; i.e.\ $(\bid^\info,\bidallocation^\info)=(\bid,\bidallocation)$ for all $\info$ and (2) $\bidallocation$ is deterministic, i.e.\ it steps up to $1$ at some $\threshold\in[0,\infty)$.
	For deterministic \bos\, define the agent's bid utility function to be $\util(d)=(\val-d)\bidallocation(d)$ for all bids $d$.
\end{definition}

We study \bos\ where the agent exactly best responds.
For deterministic \bos\ this is particularly simple to define.
We will give a more general definition when we consider randomized \bos\ in the next section.

\begin{definition}
	A deterministic single-agent \bo\ with value $\val$, deterministic bid $\bid$, and allocation rule $\bidallocation$ satisfies {\em best response} if for every alternate bid $\bid'$, $\util(\bid)\geq \util(\bid')$.
\end{definition}

\begin{lemma}\label{lem:epsres}
	The \vc\ of deterministic single-agent \bos\ satisfying best response is $\eta=1$.
\end{lemma}
\begin{proof}
	We give an argument lower bounding the individual efficiency, and then exhibit a particular example where our bound holds with equality.
	Since $\bidallocation$ is deterministic, it steps up to $1$ at some $\threshold\in[0,\infty)$. Then $\cumuthresh(z)=\threshold z$. For any $\delta>0$, the agent could bid $\threshold+\delta$ and win. Hence the best response bid $\bid$ satisfies $\util(b)\geq (\val-\threshold-\delta)$. Taking $\delta\rightarrow 0$, we obtain $\util(\bid)+\threshold\geq \val$. We may further weaken this to obtain the desired inequality for all $z\in[0,1]$: $\util(\bid)+\threshold z\geq \val z$.
	This bound is the best possible in the worst case. An example for which it holds with equality is $\val=1$, $z=1$, $\threshold = 0$, and $\bid = 0$.
\end{proof}

Note that this analysis relied heavily on the allocation rule being deterministic. We present an example demonstrating this assumption to be necessary in Section~\ref{sec:randomrule}.

\subsubsection{Randomized Rules Under Best Response}
\label{sec:bnewelfare}
\newcommand{\xgen}{X_{\text{rand}}}

Randomized environments present a richer bidding problem in two ways.
First, for a given value $\val$, the agent's information structure may induce a distribution over interim allocation rules.
Second, for a fixed allocation rule, the agent may now select from many different allocation levels and associated payments.
This latter randomness may stem from incomplete information with respect to other agent's values or from randomness in the mechanism or feasibility environment. 
Below, we give a tight single-agent analysis under exact best response, defined as follows.

\begin{definition}
	A single-agent \bo\ $(\val,\mechdist)$ satisfies {\em best response} if for every fixed deviation bid $\bid'$, we have $\util(\mechdist)\geq\expect[\info]{(\val-\bid')\allocation^\info(\bid')}.$
\end{definition}

\begin{lemma}\label{lem:brres}
	The \vc\ of single-agent \bos\ under best response is $\res=(e-1)/e$. 
	
\end{lemma}
\begin{proof}
	We give a lower bounding argument, followed by an example for which the analysis is tight. The lower bounding argument is illustrated graphically in Figure~\ref{fig:fpawelfare}.
	Define the expected allocation rule $\bidallocation(d)=\expect[\info]{\bidallocation^\info(d)}$ for all deviation bids $d$.
	Given any such deviation, the agent's utility $\util(\mechdist)$ from the \bo\ must satisfy $\util(\mechdist)\geq (\val-\dev)\bidallocation(\dev)$. 
	For any allocation probability $\allocation$, the agent could get allocation probability at least $\allocation $ by bidding arbitrarily close to $\interthresh{}{\allocation}=\inf\{\dev\,|\,\bidallocation(\dev)\geq \allocation\}$. Hence, $\util(\mechdist)\geq (\val-\interthresh{}{\allocation})\allocation$ for all $\allocation\in[0,1]$. We may rearrange this as $\interthresh{}{\allocation}\geq \val-\util(\mechdist)/\allocation$. Since we also have $\interthresh{}{\allocation}\geq 0$, we may write:
	\begin{align*}
	\cumuthreshagent{}{z}{}\geq \int_0^z \max(\val-\tfrac{\util(\mechdist)}{\allocation},0)\,d\allocation&=\int_{\util(\mechdist)/\val}^z \val-\tfrac{\util(\mechdist)}{\allocation}\,d\allocation\\
	&=\val z-\util(\mechdist)(1-\ln\tfrac{\util(\mechdist)}{\val z}).
	\end{align*}
	We therefore have $\util(\mechdist)+\cumuthreshagent{}{z}{}\geq \val z +\util(\mechdist)\ln(\util(\mechdist)/\val z)$. Holding $\val$ fixed and minimizing the righthand side as a function of $\util(\mechdist)$ yields the inequality $\util(\mechdist)+\cumuthreshagent{}{z}{}\geq \val z(e-1)/e$.  This lower bounds individual efficiency by $(e-1)/e$. For an example exhibiting the tightness of this bound, take $\val = 1$, $\bidallocation^\info(d)=\bidallocation(d)=(e(1-d))^{-1}$ for $d\in[0,(e-1)/e]$ and all $\info\in[0,1]$, and $\bid^\info = 0$.
\end{proof}

    \begin{figure}[t]
	\small
	\centering
	\subfloat[][
	\Vc\ analysis for randomized allocation rule $\bidallocation(d)=\mathbb E_{\info}{[}\bidallocation^\info(d){]}$.
	The agent's utility $\util(\mechdist)$ is at least the utility from best single bid $\bidutil(\bid)=(\val-\bid)\bidutil(\bid)$ shown as the bottom-right shaded box.
	Together, $\bidutil(\bid)$ and $\cumuthresh(z)$ approximate the dotted box, of area $\val z$.
	\label{fig:approx}]{\begin{tikzpicture}[xscale=3.5, yscale=3.5, domain=0:0.9, smooth]
			\draw  (1.12,1) -- (.91,0.7) .. controls ( 0.68,0.4) and ( 0.6,0.32) .. (0.16,0.1) --(0,0);;
			\node at (1.2,0.9) {\footnotesize $\bidallocation(\dev)$};
			\draw[-] (0,0) -- (0,1) node[left] {$1$};
			\draw[-] (0,0) -- (1.4,0) node[below] {Bid ($\dev$)};
			\node at (1, -0.09) {\footnotesize $\val$};
			\draw[pattern=north east lines, pattern color=lightgray] ( 0.5,0) rectangle ( 1,0.29);
			\node [fill=white!50] at (0.75,0.15){\footnotesize $\bidutil(\bid)$};
			\draw[ pattern=north west lines, pattern color=lightgray] (0,0) -- (0,.7) -- (1.12, .7) -- (.91,0.7) .. controls ( 0.68,0.4) and ( 0.6,0.32) .. (0.16,0.1) --(0,0);
			\node [fill=white] at ( 0.35,0.46) {\footnotesize $ \cumuthresh(z)$};
			\node at (-0.15,0.3) {\footnotesize $\bidallocation(\bid)$};
			\draw[-] (0,0) -- (0,.7) node[left] {$z$};
			\draw[] (-0.02, 0.29) -- (0, 0.29);     
			\node at (0.5,-0.09) {\footnotesize $\bid$};
			\draw[] (1, 0) -- (1, -0.025);
			\draw[] (0.5, 0) -- (0.5, -0.025);
			\draw[line width=1.3pt, dashed] (0,0) rectangle (1,.7);
			\node at (0.7,1.1) {Bid Allocation Rule};
		\end{tikzpicture}
		} \hspace*{1cm}
	\subfloat[][The agent's equiutility curve $\bidutil(\bid)/(\val-\dev)$ (dark line) upper bounds their bid-allocation rule, yielding a lower bound on $\cumuthresh(z)$.
	\label{fig:indifferent} ]{
		\begin{tikzpicture}[xscale=3.5, yscale=3.5, domain=0:0.9, smooth]
			\draw  (1.12,1) -- (.91,0.7) .. controls ( 0.68,0.4) and ( 0.6,0.32) .. (0.16,0.1) --(0,0);;
			\draw[line width=1.4pt] (0.85, 1) -- (0.82, 0.9) -- (.75,0.7) .. controls ( 0.55,0.19) and ( 0.35,0.20) .. (0,0.15);     
			\node at (0.53,0.55) {\footnotesize $ \frac{\bidutil(\bid)}{\val-\dev}$};   
			\draw[pattern=north east lines, pattern color=lightgray] ( 0.5,0) rectangle ( 1,0.29);
			\node at (1.2,0.9) {\footnotesize $\bidallocation(\dev)$};
			\node [fill=white!50] at (0.75,0.15){\footnotesize $\bidutil(\bid)$};
			\draw[-] (0,0) -- (0,1) node[left] {$1$};
			\draw[-] (0,0) -- (1.4,0) node[below] {Bid ($\dev$)};
			\node at (1, -0.09) {\footnotesize $\val$};
			\node at (-0.16,0.3) {\footnotesize $\bidallocation(\bid)$};       
			\draw[] (-0.02, 0.3) -- (0, 0.3);       
			\node at (0.5,-0.09) {\footnotesize $\bid$};
			\draw[] (1, 0) -- (1, -0.025);
			\draw[] (0.5, 0) -- (0.5, -0.025);
			\node at (-0.15, 0.15) {\footnotesize $\frac{ \bidutil(\bid)}{\val}$};
			\draw[] (-0.02, 0.155) -- (0, 0.155);
			\draw[line width=1.3pt, dashed] (0,0) rectangle (1,.7);
			\node at (.7,1.1) {Bid Allocation Rule};
		\end{tikzpicture}
	}   
	\caption{ \label{fig:fpawelfare}
	}
\end{figure}
\subsubsection{Approximate Best Response}
\label{sec:approx}

This section generalizes the analyses of \vc\ beyond the strong assumption of best response by the agent, to $(1-\epsilon)$-best response. 
We show that the best response analyses of the previous sections degrade smoothly with the level of best response.
This result shows that robust efficiency bounds using \rc\ and \vc\ are also robust to agent misoptimization.
This is especially important where equilibria of exact best response might not exist, such as with pure Nash equilbria of first-price auctions.

\begin{definition}
A single-agent \bo\ $(\val,\mechdist)$ satisfies 
{\em $(1-\epsilon)$-best response} if for every fixed deviation bid $\bid'$, we have $\util(\mechdist)\geq(1-\epsilon)\expect[\info]{(\val-\bid')\allocation^\info(\bid')}.$
\end{definition}

Approximate best response compares utility in the \bo\ to the utility of the best fixed bid.
We will analyze the \vc\ of \bos\ under approximate by best response by showing that these \bo\ inherit the \vc\ of this bid, degraded by a factor of $(1-\epsilon)$. Formally:

\begin{lemma}\label{lem:epsreduction}
Let $(\val,\bid,\bidallocation^\info)$ be a single-agent \bo\ with a fixed bid $\bid$ and a random allocation rule $\bidallocation^\info$. If $(\val,\bid,\bidallocation^\info)$ has \vc\ $\eta$, then any \bo\ $(\val, { \bid}^{\info},\bidallocation^\info)$ with random bid $\bid^\info$ satisfying $(1-\epsilon)$-best response has \vc\ at least $(1-\epsilon)\eta$.
\end{lemma}
\begin{proof}
We have the following for any target allocation $z\in[0,1]$.
\begin{align*}
	 \util({\bid}^\info,\bidallocation^\info)&\geq (1-\epsilon)\util({\bid},\bidallocation^\info)\\
	 &\geq (1-\epsilon)(\res\val z-\cumuthresh(z))\\
	 &\geq (1-\epsilon)\res\val z - \cumuthresh(z).
\end{align*}
The first inequality follows from approximate best response, and the second from the \vc\ of $\bid$.
We conclude that $(\val, { \bid}^{\info},\bidallocation^\info)$ has \vc\ at least $(1-\epsilon)\res$.
\end{proof}

We may combine Lemma~\ref{lem:epsreduction} with the best response analyses of Lemma~\ref{lem:epsres} and \ref{lem:brres} to obtain new \vc\ guarantees.
\begin{corollary}
	Let $(\val,\bid^\info,\bidallocation^\info)$ be a single-agent \bo\ satisfying $(1-\epsilon)$-best response. Then $(\val,\bid^\info,\bidallocation^\info)$ has \vc\ at least $(1-\epsilon)(e-1)/e$. If $(\val,\bid^\info,\bidallocation^\info)$ is deterministic (Definition~\ref{def:bodet}), then it has \vc\ at least $(1-\epsilon)$.
\end{corollary}
\begin{proof}
	Let $\bid$ be a fixed bid maximizing $\expect[\info]{(\val-\bid)\bidallocation(\bid)}$.
	The new \bo\ $(\val,\bid,\bidallocation^\info)$ satisfies best response and has \vc\ at least $(e-1)/e$, by Lemma~\ref{lem:brres} (resp.\ \vc\ $1$ by Lemma~\ref{lem:epsres}, if $(\val,\bid,\bidallocation^\info)$ is deterministic).
	Applying Lemma~\ref{lem:epsreduction}, we conclude that $(\val,\bid^\info,\bidallocation^\info)$ must have \vc\ at least $(1-\epsilon)(e-1)/e$ (resp.\ $(1-\epsilon)$).
\end{proof}
\subsection{Performance Guarantees}
\label{sec:perf}
\Rc\ and \vc\ together yield performance guarantees for auction equilibria.  Section~\ref{sec:welf} gives the welfare consequences. In Section~\ref{sec:welfare-bnerev}, we then use the reduction from revenue maximization to welfare maximization of \citet{M81} to show that similar guarantees apply to winner-pays-bid mechanisms with reserves when values are independent.
\subsubsection{Robust Welfare Guarantees}
\label{sec:welf}

\Vc\ was motivated in terms of the welfare properties of a single-buyer allocation problem, where a seller chooses between the buyer and an outside option. In multi-buyer mechanisms, the seller's outside options are endogenously generated by competition; \trc\ measures the way competition translates into revenue. Hence, \vc\ and \trc\ combine to bound welfare in equilibrium.
To formalize this discussion, we first isolate for each agent and realized value a single-agent \bo\ (Definition~\ref{def:bo}).
We will study the \vc\ of these \bos.
\begin{definition}
	Given equilibrium $\jointcdf$ and agent $\agent$ with value $\valagent$, the {\em conditional \bo}\ for $\agent$ is defined by $\agent$'s conditional bid distribution and conditional distribution of bid allocation rules. Formally, the distribution of $(\valothers,\bids)\sim \jointcdf\,|\,\valagent$ induces a joint distribution $\mechdist_i(\valagent)$ over bids $\bidagent$ for $\agent$ and over allocation rules $\bidallocagent(\cdot,\bidothers)$. The conditional \bo\ for $\agent$ is $(\valagent,\mechdist_\agent(\valagent))$.
\end{definition}

The \vc\ of the conditional \bos\ and the \rc\ of the mechanism together imply a robust welfare guarantee.

\begin{theorem}\label{thm:restowelf}
	Let $\mech$ be a winner-pays-bid mechanism with \rc\ $\revpar\leq 1$.
	Let $\jointcdf$ be an equilibrium for $\mech$ in which the agents' conditional \bos\ have \vc\ $\eta$. Then 
	the expected welfare in $\jointcdf$ is a $\revpar\res$-approximation to the optimal welfare.
\end{theorem}
\begin{proof}
	For any value $\valagent$ for agent $\agent$, let $\allocoptagent(\valagent)$ denote agent $\agent$'s interim allocation probability under the welfare-optimal allocation rule, and let $\utilagent(\valagent)$ denote agent $\agent$'s interim expected utility in $\jointcdf$. The individual efficiency of $\agent$'s conditional bidding outcome implies:
		\begin{equation*}
		\utilagent(\valagent)+\cumuthreshagent{\agent}{\allocoptagent(\valagent)\,|\,\valagent}{}\geq  \res\valagent\allocoptagent(\valagent).
		\end{equation*}
		Summing over all agents and taking expectation over $\vals$ yields:
		\begin{equation*}
		\expect[\vals]{\sum\nolimits_{\agent}\utilagent(\valagent)}+\expect[\vals]{\sum\nolimits_{\agent}\cumuthreshagent{\agent}{\allocoptagent(\valagent)\,|\,\valagent}{}}\geq \res\expect[\vals]{\sum\nolimits_{\agent}v_i\allocoptagent(\valagent)}.
		\end{equation*}
		The righthand side is the optimal expected welfare. The second term on the left can be rewritten as $\expect[\vals]{\sum\nolimits_{\agent}\cumuthreshagent{\agent}{\allocoptagent(\valagent)\,|\,\valagent}{}}=\sum\nolimits_{\agent}\expect[\valagent]{\cumuthreshagent{\agent}{\allocoptagent(\valagent)\,|\,\valagent}{}}$. We may therefore apply \trc\ to obtain:
		\begin{equation*}
			\expect[\vals]{\sum\nolimits_{\agent}\utilagent(\valagent)}+\tfrac{1}{\revpar}\rev(\mech,\valuecdfs,\bidcdfs)\geq \res\WEL(\OPTmech,\valuecdfs).
		\end{equation*}
		The result then follows from noting that $\revpar\leq 1$ and that the welfare of $\mech$ is the sum of the expected utilities and revenue.
\end{proof}

Using the individual efficiency guarantees of Lemmas~\ref{lem:epsres} and~\ref{lem:brres}, we can instantiate Theorem~\ref{thm:restowelf}.
Different equilibrium concepts impose different constraints on the joint distribution of bids and values.
These assumptions dictate the relevant individual efficiency guarantee.
We first state the relevant notions of exact and approximate best response for multi-agent mechanisms.
We state the definition in full generality for mechanisms which may have actions that are more complex than single bids.

\begin{definition}
	\label{def:br}
	An equilibrium $\jointcdf$ satisfies 
	{\em $(1-\epsilon)$-best response} if for each agent $\agent$, value $\valagent$, and deviation action $\action$,
	\begin{equation}\label{eq:epsbr}
		\expect[\actions\,|\,\valagent]{\bidutilagent(\actions)}\geq(1-\epsilon) \expect[\actionothers\,|\,\valagent]{\bidutilagent(\action,\actionothers)}.
	\end{equation}
	If $\jointcdf$ satisfies (\ref{eq:epsbr}) with $\epsilon=0$, we say it satisfies {\em best response}.
\end{definition}
Most well-studied equilibrium concepts satisfy the best response condition of Definition~\ref{def:br}.
For example \citet{BM16,BBM-17} impose an information structure on the game, which they model with value-dependent signals for each agent.
They assume each agent best responds to their signal.
While \citet{BBM-17} allow interdependent values, we restrict to private values.
However, the set of private-value equilbria we allow is wider, and includes the coarse analog of their Bayes correlated equilibria.
Other concepts which satisfy Definition~\ref{def:br} are Bayes-Nash equilibrium and the communication equilibria of \citet{F86}.
Combining Theorem~\ref{thm:restowelf} with Lemma~\ref{lem:brres} yields the following guarantee for exact best response:

\begin{corollary}
	Any winner-pays-bid mechanism with \rc\ $\revpar\leq 1$ has robust welfare approximation $\revpar(e-1)/e$ for equilibria satisfying best response.
\end{corollary}

Lemma~\ref{lem:epsres} yields the following improved bound in pure $\epsilon$-Nash equilibrium, where bids and value are deterministic and satisfy $(1-\epsilon)$-best response.

\begin{corollary}\label{cor:nash}
Any winner-pays-bid mechanism with deterministic allocation rule and \rc\ $\revpar\leq 1$ has robust welfare approximation $\revpar(1-\epsilon)$ in $\epsilon$-Nash equilibrium.
\end{corollary}


\subsubsection{Revenue Analysis}
\label{sec:welfare-bnerev}

\Vc\ quantifies the way a buyer chooses to trade off their utility and the seller's utility (via \trc) against their value. This section extends these ideas to study the objective of seller revenue under suitable independence conditions on values and bids. With carefully chosen reserve prices, we will obtain a robust revenue approximation of $\revpar(e-1)/2e$ for winner-pays-bid mechanisms with \rc\ $\revpar$.
We consider equilibria satisfying the following independence condition, which we state in generality beyond single-bid mechanisms.
\begin{definition}\label{def:nocom}
	An equilibrium $\jointcdf$ satisfies {\em no bidder communication} if for every agent $\agent$ and value $\valagent$, $\actionagent$ is independent of $\valothers$ conditioned on $\valagent$.
\end{definition}
Definition~\ref{def:nocom} rules out equilibria where bidders communicate nontrivial information about their types to each other or a mediator, and hence eliminates many Bayes correlated equilibria \citep{BM16,BBM-17} and communication equilibria \citep{F86}. It is important to rule these out, as it is possible to construct low-revenue equilbria under these notions of correlation.
Definition~\ref{def:nocom} still permits some correlated equilibria, however, namely those with one-way communication from a mediator.
In particular, Definition~\ref{def:nocom} still permits coarse correlated equilibria of the agent-strategic or agent-normal forms of the Bayesian game, and hence learning outcomes under incomplete information \citep{HST15} as well as standard Bayes-Nash equilibrium.

Given an equilibrium $\jointcdf$ with independently distributed values and no bidder communication, the analysis of \citet{M81} implies that the ex ante expected payment of an agent $\agent$ is $\expect[\valagent]{\vvalagent(\valagent)\allocagent(\valagent)}$, where
$\vvalagent(\valagent) = \valagent - \frac{1-F_i(\valagent)}{\valuepdf(\valagent)}$ is the {\em virtual value} for value $\valagent$
and $F_i$ (resp.\ $f_i$) the cumulative distribution function (resp.\ probability density function) of $\agent$'s value distribution. It follows that $\rev(\mech,\jointcdf) =\expect[\vals,\bids]{\sum_\agent\vvalagent(\valagent)\bidallocagent(\bids)}$. We refer to an atomless distribution with $\vvalagent(\valagent)$ nondecreasing in $\valagent$ as {\em regular}. For regular distributions, the revenue-optimal mechanism chooses the allocation with the highest virtual surplus $\sum_\agent\vvalagent(\valagent)\allocagent$. For downward-closed settings, agents with negative virtual value are excluded via {\em monopoly reserves}, given by $\reserveagentm=\inf\{\valagent\,|\,\vvalagent(\valagent)\geq 0\}$.

The above characterization rewrites equilibrium revenue in terms of buyer surplus in a transformed value space. Revenue maximization then amounts to excluding agents with negative virtual values, and maximizing virtual welfare among those that remain. We will extend the \vc\ analysis of randomized allocation rules (Lemma~\ref{lem:brres}) to reason about an agent's contribution to revenue under monopoly reserves.

We focus on equilibria and corresponding \bos\ where reserves successfully exclude low-valued agents.
In particular, we rule out equilbria where agents overbid to beat the reserve, and those where high-valued agents could attain positive utility but bid below the reserve inappropriately.
Furthermore, we make the standard tie-breaking assumption that agents who are indifferent between bidding below the reserve and above it bid above.
These properties hold under many standard solution concepts, including Bayes Nash equilibrium, where they follow from standard best response arguments.
\begin{definition}
	An equilibrium $\jointcdf$ for mechanism $\mech$ with reserves $\reserves$ {\em respects reserves} if for all agents $\agent$, $\valagent\geq \reserveagent$ if and only if $\bidagent \geq \reserveagent$. Similarly, a single-agent \bo\ $(\val,\mechdist)$ respects reserve $\reserve$ if $\bid^\info\geq \reserve$ with probability $1$ if and only if $\val\geq \reserve$.
\end{definition}

When a \bo\ respects the monopoly reserve price $\reservem$, we can extend the analysis of \vc\ to virtual welfare.

\begin{lemma}\label{lem:vvalres}
	Let $\singlecdf$ be a regular value distribution with monopoly reserve $\reservesingm$.
	Further let $(\val,\bid^\info,\bidallocation^\info)$ be a single-agent \bo\ that respects $\reservem$.
	If $(\val,\bid^\info,\bidallocation^\info)$ further  satisfies best response, then for any target allocation $z\in[0,1]$,
	\begin{equation}\label{eq:vvalres}
	\mathbb E_{\info}[\vval(\val)\bidallocation^\info(\bid^\info)]+\cumuthreshagent{}{z}{\reservesingm}\geq \tfrac{e-1}{e}\vval(\val)z.
	\end{equation}
\end{lemma}
\begin{proof}
	The argument is illustrated geometrically in Figure~\ref{fig:covering}.
	Since $F$ is regular, $\vval(\val)<0$ if and only if $\val<\reservem$, in which case the agent bids below $\reservem$ almost surely. For such agents, (\ref{eq:vvalres}) holds trivially.
	For an agent with value $\val\geq\reservesingm$, bids are at least $\reservesingm$. Hence given  $\bid(\val)$, we have $\mathbb E_\info[\bidpay(\bid^\info)]\geq \reservesingm\bidallocation(\reservesingm)$. Furthermore, if $\allocation\leq \bidallocation(\reservesingm)$, then $\interthresh{}{\allocation}{}\leq \reservesingm$. It follows that $\int_0^{\bidallocation(\reservesingm)}\interthresh{}{\allocation}{}\,d\allocation\leq \reservesingm\bidallocation(\reservesingm)$. We therefore obtain the following sequence of inequalities:
\begin{equation}\label{eq:paymentvsthresh}
	\mathbb E_\info[\bidpay(\bid^\info)]\geq \reservesingm\bidallocation(\reservesingm)\geq \int_0^{\bidallocation(\reservesingm)}\interthresh{}{\allocation}{}\,d\allocation\geq \cumuthreshagent{}{z}{}-\cumuthreshagent{}{z}{\reservesingm}.
\end{equation}
This inequality is illustrated in Figure~\ref{fig:welfareboundgeo}.
	Combining with Lemma~\ref{lem:brres}, we obtain a tradeoff between an agent's contribution to surplus and their threshold surplus:
	\begin{equation}\label{eq:valcoverres}
	\mathbb E_{\info}[\val\bidallocation^\info(\bid^\info)]+\cumuthreshagent{}{z}{\reservesingm}= \util(\mechdist)+\mathbb E_\info[\bidpay(\bid^\info)]+\cumuthreshagent{}{z}{\reservesingm}\geq  \util(\mechdist)+\cumuthreshagent{}{z}{}\geq \tfrac{e-1}{e}\val z.
	\end{equation}
	Inequality (\ref{eq:valcoverres}) can be seen in Figure~\ref{fig:welfareboundgeo} as well.
	The lemma then follows from noting that $\vval(\val)\leq \val$, and hence when $\vval(\val)\geq 0$, the inequality (\ref{eq:vvalres}) is a weakening of (\ref{eq:valcoverres}). The resulting inequality (\ref{eq:vvalres}) is illustrated geometrically in Figure~\ref{fig:revenueboundgeo}.
\end{proof}

\begin{figure}[t]
	\small
	\centering
	\subfloat[][Equations (\ref{eq:paymentvsthresh}) and (\ref{eq:valcoverres}), illustrated for a deterministic bid $\bid^\info=\bid$. Payments $\bidpay(\bid)=\bid\bidallocation(\bid)$ (thick gray rectangle) cover the threshold surplus from the reserve, given by  $\cumuthreshagent{}{z}{}-\cumuthreshagent{}{z}{\reservesingm}$ (dotted box). Combining this with individual efficiency of best response, we see that the shaded areas and $\bidpay(\bid)$ together cover a $(e-1)/e$ fraction of dashed box, of area $\val z$.  \label{fig:welfareboundgeo} ]{
		\begin{tikzpicture}[xscale=3.5, yscale=3.5, domain=0:0.9, smooth]
			\draw[-] (0,0) -- (0,.7) node[left] {$z$};
			\draw[-] (0,0) -- (0,1) node[left] {$1$};
			\draw  (1.12,1) -- (.91,0.7) .. controls ( 0.68,0.4) and ( 0.6,0.32) .. (0.4,0.2) --(.4,0);
			\draw (0,.2) -- (.4,.2);
			\node at (1.2,0.9) {\footnotesize $\bidallocation(\dev)$};
			\draw[-] (0,0) -- (1.4,0) node[below] {Bid ($\dev$)};
			\draw[line width=1.3pt, dashed] (0,0) rectangle (1,.7);

			\node at (1, -0.09) {\footnotesize $\val$};
			\draw[pattern=north east lines, pattern color=lightgray] ( 0.53,0) rectangle ( 1,0.29);

			\node [fill=white!50] at (0.76,0.15){\footnotesize $\bidutil(\bid)$};
			\draw[ pattern=north west lines, pattern color=lightgray] (0,.2) -- (0,.7) -- (1.12, .7) -- (.91,0.7) .. controls ( 0.68,0.4) and ( 0.6,0.32) .. (0.4,0.2) -- (0,0.2) ;\draw[pattern=dots, pattern color=lightgray] ( 0,0) rectangle ( .4,0.2);
			\node [fill=white] at ( 0.35,0.46) {\footnotesize $ \cumuthresh(z)$};
x
	\node at (-0.17,0.31) {\footnotesize $\bidallocation(\bid)$};
	\node at (-0.16,0.20) {\footnotesize $\bidallocation(\reservesingm)$};
						\draw[line width=0.6mm, color = lightgray] (0,0) rectangle (.53,.29);
			\draw[] (-0.02, 0.29) -- (0, 0.29);  
			\draw[] (-0.02, 0.2) -- (0, 0.2);    
			\node at (0.5,-0.09) {\footnotesize $\bid$};
			\draw[] (1, 0) -- (1, -0.025);
			\draw[] (0.5, 0) -- (0.5, -0.025);
			\node at (0.7,1.1) {Bid Allocation Rule};
		\end{tikzpicture}
	}\IFECELSE{\hspace*{0.2cm}}{\hspace*{1cm}}
	\subfloat[][Equation (\ref{eq:vvalres}), illustrated for a deterministic bid $\bid^\info=\bid$. Lemma \ref{lem:vvalres} shows the shaded areas cover an $(e-1)/e$ fraction of the dashed box, which has area $\vval(\val)z$. Note that $\cumuthreshagent{}{z}{\reservesingm}$ omits the area below $\bidallocation(\reservesingm)$, and that the boxes for $\bidutil(\bid)$ and $\vval(\val)\bidallocation(\bid)$ have the same height. \label{fig:revenueboundgeo}]{
		\begin{tikzpicture}[xscale=3.5, yscale=3.5, domain=0:0.9, smooth]

			\draw[] (-0.02, 0.2) -- (0, 0.2);
			\draw[ pattern=north west lines, pattern color=lightgray] (0,.1) -- (0,.7) -- (1.12, .7) -- (.91,0.7) .. controls ( 0.68,0.4) and ( 0.6,0.32) .. (0.4,0.2) --(0,0.2);
			\node [fill=white] at ( 0.35,0.46) {\footnotesize $ \cumuthreshagent{}{z}{\reservesingm}$};
			\draw  (1.12,1) -- (.91,0.7) .. controls ( 0.68,0.4) and ( 0.6,0.32) .. (0.4,0.2) --(0.0,.2);
			\node at (1.2,0.9) {\footnotesize $\bidallocation(\dev)$};
			\draw[-] (0,0) -- (0,.7) node[left] {$z$};
			\draw[-] (0,0) -- (0,1) node[left] {$1$};
			\draw[-] (0,0) -- (1.4,0) node[below] {Bid ($\dev$)};
			\draw[] (1, 0) -- (1, -0.025);      
			\node at (1, -0.09) {\footnotesize $\val$};
			\draw[] (0.75, 0) -- (0.75, -0.025);        
			\node at (0.75, -0.09) {\footnotesize $\vval(\val)$};
			\draw[pattern=north east lines, pattern color=lightgray] 
			( 0,0) rectangle ( 0.75, 0.29);
			\node [fill=white!50] at (1.1,0.18){\footnotesize $\vval(\val)\bidallocation(\bid)$};
			\draw[] (.5,.14) -- (.85,.18);
			\node at (-0.17,0.31) {\footnotesize $\bidallocation(\bid)$};
			\node at (-0.16,0.20) {\footnotesize $\bidallocation(\reservesingm)$};  
			\draw[] (-0.02, 0.29) -- (0, 0.29); 
			
			\node at (0.5,-0.09) {\footnotesize $\bid(\val)$};
			\draw[] (0.5, 0) -- (0.5, -0.025);
			\draw[line width=1.5pt, dotted] (0,0) rectangle (0.75,.7);
			\node at (0.9,1.1) {Bid Allocation Rule};
		\end{tikzpicture}
	}   
	\caption{ \label{fig:covering}}
\end{figure}

We may apply Lemma~\ref{lem:vvalres} to obtain a revenue guarantee.

\begin{theorem}\label{thm:bnerev}
	Let $\jointcdf$ be a no-bidder-communication equilibrium of a winner-pays-bid mechanism $\mech$ with monopoly reserves $\reservesm$. Assume values are independently distributed according to regular distributions, and that $\jointcdf$ satisfies best response and respects $\reservesm$. Then the expected revenue is a $\revpar(e-1)/2e$-approximation to that of the optimal mechanism.
\end{theorem}
\begin{proof}
	Fix a value $\valagent$ for agent $\agent$. If $\valagent\geq\reserveagentm$, Lemma~\ref{lem:vvalres} implies
	\begin{equation}\label{eq:vvalcoveropt}
	\vvalagent(\valagent)\allocagent(\valagent)+\cumuthreshagent{\agent}{\allocoptagent(\valagent)}{\reserveagentm}\geq \tfrac{e-1}{e}\vvalagent(\valagent) \allocoptagent(\valagent),
	\end{equation}
	where $\allocoptagent(\valagent)$ denotes the interim allocation probability for $\agent$ under allocation rule of the revenue-optimal mechanism $\OPTmech_\valuecdfs$ for $\valuecdfs$, and $\allocagent(\valagent)$ denotes their interim allocation probability in equilibrium.
	We omit conditioning from $\cumuthreshagent{\agent}{\allocoptagent(\valagent)}{\reserveagentm}$, as $\valuecdfs$ is a product distribution. We may sum (\ref{eq:vvalcoveropt}) over all agents and take expectations to obtain
	\begin{equation*}
	\expect[\vals]{\sum\nolimits_{\agent}\vvalagent(\valagent)\allocagent(\valagent)}+\expect[\vals]{\sum\nolimits_{\agent}\cumuthreshagent{\agent}{\allocoptagent(\valagent)}{\reserveagentm}}\geq \tfrac{e-1}{e}\expect[\vals]{\sum\nolimits_{\agent}\vvalagent(\valagent) \allocoptagent(\valagent)}.
	\end{equation*}
	Applying \trc\ and noting that $\revpar\geq 1$ yields:
	\begin{equation*}
	\tfrac{1}{\revpar}\Big(\expect[\vals]{\sum\nolimits_{\agent}\vvalagent(\valagent)\allocagent(\valagent)}+\rev(\mech^{\reservesm},\jointcdf)\Big)\geq \tfrac{e-1}{e}\expect[\vals]{\sum\nolimits_{\agent}\vvalagent(\valagent) \allocoptagent(\vals)}.
	\end{equation*}
	Since a mechanism's expected revenue is equal to its expected virtual surplus, we obtain the desired revenue guarantee:
	\begin{equation*}
	2\rev(\mech^{\reservesm},\jointcdf)\geq \revpar\tfrac{e-1}{e}\rev(\OPTmech_\valuecdfs,\valuecdfs).\qedhere
	\end{equation*}
\end{proof}
\subsection{Lower Bounds}
\label{sec:welfare-lbs}
\label{sec:lb}

To conclude the section, we explore the limits of our approach. We do so by exhibiting three sets of examples. In Section~\ref{sec:singlelb}, we describe the worst-known examples for welfare and revenue loss in the single-item first-price auction in Bayes-Nash equilibrium. Section~\ref{sec:randomrule} gives a Nash equilibrium of a mechanism with randomized allocations, \rc\ $1$, and a factor of $(e-1)/e$ welfare loss. This shows that the restriction to deterministic allocation rules in our analysis of $\epsilon$-Nash equilibrium was necessary. Finally, \citet{DK15} study the extent to which low \rc\ is a necessary condition for a good robust welfare guarantee. We briefly discuss their results in Section~\ref{sec:welfare-nec}.

\subsubsection{Single-Item Lower Bounds}
\label{sec:singlelb}

The \rc\ and individual efficiency provide an general framework for robust welfare and revenue analysis. Whether the resulting welfare guarantees are best possible will vary across mechanisms and solution concepts. What follows are three equilibria of the single-item first-price auction which are the worst known for their type. We begin with a Bayes-Nash equilibrium with correlated values which matches Theorem~\ref{thm:restowelf} exactly, due to \citet{S14}.

\begin{example}\label{ex:cor}
	The  example has three agents. Agents $1$ and $2$ have values perfectly correlated, and drawn according to the distribution with CDF $(e(1-v))^{-1}$ for $v\in[0,1-1/e]$. Agent $3$ has value deterministically $1$. If we break ties in favor of agent $3$, it is a Bayes-Nash equilibrium for agents $1$ and $2$ to bid their value and agent $3$ to bid $0$. In the optimal allocation, agent $3$ always wins. A straightforward calculation shows the equilibrium welfare to be $(e-1)/e$.
\end{example}

Example~\ref{ex:cor} provides a canonical setting in which Theorem~\ref{thm:restowelf}.
When values are independently distributed, however, \citet{JL22} exploit independence to improve the welfare approximation guarantee from $(e-1)/e \approx .63$ to $(e^2-1)/e^2\approx.865$. Furthermore, they exhibit an example where this bound is tight. As in Example~\ref{ex:cor}, the tight example has a single high-valued agent, who should always receive the item in the welfare-optimal allocation, but underbids in equilibrium.
Without correlation to coordinate the bid distribution of the low-valued competition, however, \citet{JL22} instead use a large population of identical, independent agents to implement a competing distribution for the high-valued agent.

Finally, we give the worst-known example for revenue in Bayes-Nash equilibrium with monopoly reserves, with independent, regularly-distributed values. As with welfare under independence, there is a gap between the example and the robust guarantee of Theorem~\ref{thm:bnerev}.

\begin{example}
	The equilibrium will have two agents. The first agent will have value deterministically $1$, and the second will have value CDF $F_2(\val_2)=1-1/\val_2$. The monopoly reserve for agent $1$ is trivially $1$. For agent $2$, any price above $1$ has virtual value $0$, so the monopoly reserve is ambiguous. Perturb $F_2$ slightly such that the monopoly reserve is $1$. In equilibrium, agent $1$ bids their value, $1$. If we break ties in favor of agent $2$, then agent $2$ also bids $1$, yielding an expected revenue of $1$. This obtains less revenue than posting a price of $H>1$ to agent $2$, and upon rejection selling to agent $1$, which has expected revenue $2-1/H$. Taking $H\rightarrow\infty$ yields the desired multiplicative loss of $1/2$.
\end{example}

\subsubsection{Welfare Loss Under Randomized Mechanisms}
\label{sec:randomrule}


Section~\ref{sec:welfare-epsne} gave an improved \vc\ guarantee for deterministic allocation rules when compared to randomized rules. This translated to an improved welfare guarantee. We now demonstrate that the restriction to deterministic rules is necessary for this improvement. We do so by giving a mechanism with \rc\ $1$, randomized allocations, and Nash equilibrium with multiplicative welfare loss $(e-1)/e$. 

\begin{example}
Define a {\em partial allocation environment} to be an $\numagents$-agent feasibility environment given by a vector $(z_1,\ldots, z_{\numagents})$ of maximum allocations. A feasible allocation selects an agent $\agent$ and assigns them up to $z_{\agent}$ units of allocation. Given a value profile $\vals$, the welfare-optimal allocation in a partial allocation environment selects the agent maximizing $\valagent z_{\agent}$. The highest bids win allocation rule takes a profile of bids $\bids$ and allocates the agent maximizing $\bidagent z_{\agent}$. This mechanism has \rc\ $1$. 

Now consider the following convex combination of partial allocation environments. Draw a parameter $\theta$ according to $G(\theta)=(e(1-\theta))^{-1}$, and consider the three-agent partial allocation environment with maximum allocations $(\theta,\theta,1)$. By Lemma~\ref{lem:convex-combination} \trc\ is closed under convex combination, so the winner-pays-bid, highest-bids-wins mechanism for this environment has \rc\ $1$. We will now construct a Nash equilibrium for this mechanism with welfare approximation $e/(e-1)$. Let values be $(1,1,1)$. If we allocate to break ties to favor agent $3$, then it is a Nash equilibrium for agents $1$ and $2$ to bid $1$, and agent $3$ to bid $0$. Equilibrium welfare is therefore the expected value of $\theta$, which is $(e-1)/e$, rather than the $1$ that could be achieved by allocating agent $3$.
\end{example}
\subsubsection{Necessity of High \RC}
\label{sec:welfare-nec}

We have shown that high \rc\ is a sufficient condition for robust performance guarantees. We now briefly examine whether it is also a necessary condition. In other words, we study the extent to which a mechanism with poor \rc\ must also possess equilibria with welfare far from optimal. A formal treatment of this question appears in \citet{DK15}. For brevity, we simply overview the main ideas.

In Section~\ref{sec:rc-hbw}, we exhibited the highest-bids-win rule for single-minded combinatorial auctions as a mechanism with poor \rc. We gave a bundle structure with $\numitems$ items and an agent $\agent$ desiring item $\agent$ for each $\agent\in\{1,\ldots,\numitems\}$, along with an agent $\numitems+1$ who desires the grand bundle. The bid profile $(0,\ldots, 0,1)$ then refuted any \rc\ better than $\numitems$. Note, however, that under the winner-pays-bid format, this bid profile is not a Nash equilibrium for any positive value of agent $\numitems+1$, as this agent would always prefer to lower their bid. If we allow ourselves to augment the setting with an additional bidder, however, we may produce an equilibrium with welfare approximation $\numitems$. Specifically, add a second grand bundle bidder, $\numitems+2$. The value profile $(1,\ldots, 1)$ and bid profile $(0,\ldots, 0, 1, 1)$ is a Nash equilibrium for any tiebreaking: the only unilateral deviations in which bidders $1,\ldots, \numitems$ can win involve bidding at least their value, and if bidders $\numitems+1$ or $\numitems+2$ bid less than $1$, they are guaranteed to lose.

The approach of duplicating bidders can be extended to most winner-pays-bid mechanisms of interest for single-minded combinatorial auctions. In particular, consider any winner-pays bid mechanism $\mech$ which, given two or more bidders with identical desired bundles, always allocates the one with the highest bid. Let bid profile $\bids$ and alternate allocation $\allocalt$ exhibit a \rc\ at most $\revpar$. For each winner $\agent$ under $\bids$ in $\mech$, add a duplicate $\agent'$ desiring the same bundle. The following value and bid profiles are then a Nash equilibrium with welfare approximation $\revpar$ in the augmented setting: for each winner $\agent$ under $\bids$, give $\agent$ and $\agent'$ both value and equilibrium bid equal to $\bidagent$. For each agent $i$ winning under $\allocalt$ but not winning under $\bids$, give that agent value $\thresholdagent(\bidothers)$ and bid $\bidagent$. Give all other agents value and bid $0$. Under appropriate tiebreaking, this is an equilibrium for the same reasons as in the previous example: winners (and their duplicates) under $\bidagent$ cannot reduce their bids without losing, and losers under $\bids$ cannot win without overbidding. Moreover, the welfare approximation is equal to the \rc\ exhibited by $\bids$ and $\allocalt$, $\revpar$.

\citet{DK15} formalize the above approach for arbitrary single-parameter environments. For any setting which can be augmented with duplicates, and any winner-pays-bid mechanism which handles duplicates sensibly, any instance with \rc\ $\revpar$ implies the existence of a related instance and equilibrium for that related instance with welfare approximation $\revpar$. For single-minded combinatorial auctions in particular, this implies that up to a constant factor, the greedy mechanism discussed in Section~\ref{sec:rc-greedy} is the optimal winner-pays-bid mechanism, with respect to the objective of robust welfare approximation. It also pinpoints the highest-bids-wins mechanism's mismanagement of inter-bidder competition as the source of its worst-case inefficiency. More generally, it further justifies the study of \trc\ as a design objective in itself.
\section{Beyond Winner-Pays-Bid Mechanisms}
\label{sec:beyondpyb}

Thus far, we have considered only winner-pays bid mechanisms, where competition is measured by threshold bids.
In this section, we generalize \trc\ to other measures of competition and families of mechanisms.
We provide two applications.
First in Section~\ref{sec:simultaneous} we show that \trc\ is preserved in the sale of identical items via simultaneous auctions to unit-demand agents.
Agents may participate in any number of auctions, and are served if they win at least one. 
If all component mechanisms have \rc\ $\revpar$, we show that their simultaneous composition must as well.
Then in Section~\ref{sec:allpay}, we analyze all-pay auctions.
For both applications, we will only consider product distributions over values; along the way, we will give a corresponding, adapted definition of \trc.

\subsection{Generalized Framework}

The robust analysis of winner-pays-bid mechanisms consisted of three steps.
First, we quantified an agent's competition.
Second, \trc\ measured how well a mechanism managed this competition across agents, independent of incentives. 
Finally, \vc\ related this competition to each agent's bidding problem.
In general single-parameter mechanisms, each agent $\agent$ takes an action $\actionagent$, which may be a richer object than a real-valued bid, and payments may not follow the winner-pays-bid format.

In winner-pays-bid mechanisms, we measured competition with threshold  bids: when a high bid was necessary to secure allocation, competition was stronger.
An agent's interim allocation rule traced out their Pareto frontier between bids and allocation probabilities.
For general mechanisms, we will consider the frontier between allocation and a real-valued cost function $\equivbidact{\actionagent\,|\,\valagent}$, which will be chosen based on the family of mechanisms being considered.
We measure competition by the cost necessary to secure different allocations:

\begin{definition}\label{def:equivthresh}
	Given joint distribution $\jointcdf$ over actions and values, value $\valagent$ for agent $\agent$, target allocation probability $\allocation\in[0,1]$, and cost function $\compagent$, agent $\agent$'s \emph{interim threshold cost} is given by $\equivthresh{\agent}{\allocation\,|\,\valagent}=\inf_{\actionagent\,:\,\bidallocagent(\actionagent\,|\,\valagent)\geq \allocation} \equivbidact{\actionagent\,|\,\valagent}$. 
\end{definition}

For winner-pays-bid mechanisms, a natural cost function for a bid $\bidagent$ is the bid itself, i.e.\ $\equivbidact{\bidagent\,|\,\valagent}=\bidagent$.
For general mechanisms, an analogue is the {\em price per unit of allocation}.
Define for action $\actionagent$ the interim allocation $\bidallocagent(\actionagent\,|\,\valagent)=\mathbb E_{\actionothers\,|\,\valagent}[\bidallocagent(\actionagent,\actionothers\,|\,\valagent)]$ and interim payment $\bidpaymentagent(\actionagent\,|\,\valagent)=\mathbb E_{\actionothers\,|\,\valagent}[\bidpaymentagent(\actionagent,\actionothers\,|\,\valagent)]$.
These yield a cost function $\equivbidagent^{\text{PPU}}(\actionagent\,|\,\valagent)$, defined below, which generally depends on $\jointcdf$, and which generalizes the winner-pays-bid analysis. In winner-pays-bid mechanisms, it will hold that $\equivbidagent^{\text{PPU}}(\bidagent\,|\,\valagent)=\bidagent$.
\begin{definition}
	Given joint distribution $\jointcdf$ over actions and values, the {\em price per unit} for  action $\actionagent$, denoted $\equivbidagent^{\text{PPU}}(\actionagent\,|\,\valagent)$, is given by $\equivbidagent^{\text{PPU}}(\actionagent\,|\,\valagent)=\bidpaymentagent(\actionagent\,|\,\valagent)/\bidallocagent(\actionagent\,|\,\valagent)$.
\end{definition}

General single-parameter mechanisms need not support a natural notion of reserve price.
For those that do, we discount costs as we discounted bids below the reserves.
This will enable us to study simultaneous first-price auctions with individualized reserves: if agent $\agent$ faces a reserve of $\reserveagent$ in each mechanism, this imposes a minimum price per unit of $\reserveagent$ for allocation overall.

\begin{definition}\label{def:reservecost}
Given joint distribution $\jointcdf$ over actions and values, desired allocation probability $\allocation\in[0,1]$, and minimum cost $\reserveagent$, an agent $\agent$'s {\em interim threshold cost with discounted reserve} is given by $\equivthreshres{i}{\allocation\,|\,\valagent}{\reserveagent}=\equivthresh{i}{\allocation\,|\,\valagent}$ if $\equivthresh{i}{\allocation\,|\,\valagent}\geq \reserveagent$ and $\equivthreshres{i}{\allocation\,|\,\valagent}{\reserveagent}=0$ otherwise.
\end{definition}

For cost functions $\equivbid_\agent$ for each agent $\agent$, we aggregate threshold costs mimicing Definition~\ref{def:intcumulative}.

\begin{definition}\label{def:cumulativeprice}
	Given joint distribution $\jointcdf$ over actions and values, allocation probability $\allocaltagent\in[0,1]$, and minimum cost $\reserveagent$, agent $\agent$'s \emph{generalized threshold surplus with discounted reserve} for $\allocaltagent$ is given by $\equivcumulative{\agent}{\allocaltagent\,|\,\valagent}{\reserveagent}=\int_0^{\allocaltagent}\equivthreshres{\agent}{\allocation\,|\,\valagent}{\reserveagent}\,d\allocation$. 
	When $\reserveagent=0$, we omit $\reserveagent$ and simply write $\equivcumulative{\agent}{\allocaltagent\,|\,\valagent}{}$.
\end{definition}

The generalized definition of \trc\ compares the threshold surplus to the mechanism's revenue.
The definition parallels Definition~\ref{def:exprc}, but threshold surplus is now taken with respect to general cost functions.
\begin{definition}\label{def:genrc}
	The generalized \rc\ of a mechanism $\mech$ and minimum costs $\reserves$ is the largest $\revpar$ such that, for any joint distribution $\jointcdf$ over values and actions and any feasible profile $\allocalt$ of interim allocation functions function, the revenue is at least a $\revpar$ fraction of the threshold surplus:
	\begin{equation*}
		\rev(\mech,\jointcdf)\geq \revpar\sum\nolimits_\agent \expect[\valagent\sim\valuecdf]{\equivcumulative{\agent}{\allocaltagent(\valagent)\,|\,\valagent}{\reserveagent}}.
	\end{equation*}
\end{definition}

For winner-pays-bid mechanisms, the price per unit of a bid is the bid itself. 
Thus with cost functions $\equivbid_\agent^{\text{PPU}}$, Definitions~\ref{def:equivthresh}, \ref{def:reservecost}, \ref{def:cumulativeprice}, and \ref{def:genrc} coincide with their analogs from Section~\ref{sec:rc}.
Hence:

\begin{lemma}
	For winner-pays-bid mechanism $\mech$ and minimum costs $\reserves$, if $\mech$ has \rc\ $\revpar$, then for price-per-unit costs, $\mech$ also has generalized \rc\ $\revpar$.
\end{lemma}

The two main applications of this section are simultaneous mechanisms and all-pay mechanisms.
Each of these two families presents an obstacle that we overcome by weakening the definitions of \rc\ and \vc.
First, recall that a winner-pays-bid mechanism's \rc\ was the same for all priors, including those with correlation. 
However, simultaneous auctions are known to possess inefficient equilibria under correlation \citep{FFGL13}. 
We therefore consider only distributions with independent values and no bidder communication (Definition~\ref{def:nocom}).

The second problem is that all-pay mechanisms are not individually efficient.
Definition~\ref{def:vc} of \vc\ required an agent to trade off utility and threshold surplus efficiently to obtain any desired allocation level $z$.
For all-pay mechanisms, only $z=1$ can be obtained efficiently in this way.
We therefore consider ex post feasible allocations $\allocalt\in\feasible\subseteq\{0,1\}^n$ when we weaken \rc, and deterministic allocations when we weaken \vc.

The restrictions discussed above enable simpler notation for the expected threshold surplus $\expect[\valagent\sim\valuecdf]{\equivcumulative{\agent}{\allocaltagent(\valagent)\,|\,\valagent}{\reserveagent}}$ from Definition~\ref{def:genrc}.
With independent values and no bidder communication, an  agent's threshold surplus does not depend on their value $\valagent$. 
We may therefore omit the conditioning on $\valagent$.
Second, with ex post feasible allocations $\allocalt\in\{0,1\}^\numagents$, we may omit the expectation over $\valagent$.
Taking these two modifications into account, we have the following:

\begin{definition}\label{def:weakrc}
		The weak \rc\ of a mechanism $\mech$ for minimum costs $\reserves$ is the largest $\revpar$ such that, for bid distribution $\bidcdfs$, product distribution $\valuecdfs$, and feasible allocation $\allocalt\in\feasible$, the revenue is at least a $\revpar$ fraction of the threshold surplus:
	\begin{equation*}
		\rev(\mech,\bidcdfs,\valuecdfs)\geq \revpar\sum\nolimits_\agent\equivcumulative{\agent}{\allocaltagent}{\reserveagent}.
	\end{equation*}
\end{definition}

The following is then immediate from definitions.

\begin{lemma}
	If a mechanism $\mech$ has generalized \rc\ $\revpar$, then it also has weak \rc\ $\revpar$.
\end{lemma}

The weakened definition of \vc\ is below. It requires the agent to efficiently trade off generalized threshold surplus $\equivcumulative{}{\cdot}{}$ and utility. In contrast to Definition~\ref{def:vc}, it only considers $z=1$. We first generalize single-agent \bos\ beyond winner-pays-bid mechanisms.

\begin{definition}
	A single-agent {\em \ao}\ is given by a value $\val\geq 0$ and a joint distribution $\mechdist=(\action^\info,\bidallocation^\info,\bidpay^\info)$ over bids, allocation rules, and payment rules, indexed uniformly by $\info$. The agent's utility is given by $\util(\mechdist)=\mathbb E_{\info}[\val \bidallocation^\info(\action^\info)-\bidpay^\info(\action^\info)]$. The agent's {\em expected allocation rule} at bid $\action$ is $\bidallocation(\action)=\expect[\info]{\bidallocation^\info(\bid)}$, and given a cost function $\equivbid$, the interim threshold costs are $\equivthreshres{}{\allocation}{}=\inf\{\equivbid(\action)\,|\,\bidallocation(\action)\geq x\}$. The threshold surplus is $\cumuthresh(z)=\int_0^z\equivthreshres{}{\allocation}{}\,dx$.
\end{definition}

\begin{definition}\label{def:weakvc}
	Let $(\val,\mechdist)$ be a single-agent \ao. The {\em weak \vc}\ of $(\val,\mechdist)$ is given by the ratio $\eta=(\util(\mechdist) + \equivcumulative{}{1}{})/\val$.
	The weak \vc\ of a family of \aos\ is the smallest weak \vc\ in the family.
\end{definition}

\begin{SCfigure}
	{\begin{tikzpicture}[scale = 4]
		
		\fill[lightgray] (0,0) -- (0.2,0) -- (0.2,0.2) -- (0.3,0.2) -- (0.3,0.3) -- (0.69,0.3) -- (0.69,0.45) -- (0.845,0.45) -- (0.845,0.70) -- (1,0.70) -- (1.0,0.8) -- (1.15,0.8) -- (1.15,1) -- (1.3,1) -- (0,1) -- cycle;
		
		\draw[-] (0,0) -- (0.2,0) -- (0.2,0.2) -- (0.3,0.2) -- (0.3,0.3) -- (0.69,0.3) -- (0.69,0.45) -- (0.845,0.45) -- (0.845,0.70) -- (1,0.70) -- (1.0,0.8) -- (1.15,0.8) -- (1.15,1) -- (1.3,1);
		
		\node at (0.5,0.75) {$\equivcumulative{}{z}{}$};
		\node at (1.15,0.65) {$\hat \allocation(\bid)$};
		
		\foreach \point in {(0,0), (0.2,0.2), (0.3,0.3), (0.69,0.45), (0.845,0.70), (1.0,0.8), (1.15,1), (0.5,0.2), (0.6,0.1), (1,0.5), (0.35,0.15), (1.05,0.1), (1.15,0.3), (0.9,0.4)}
		\filldraw[black] \point circle (.5pt);
		
		\draw[-, line width = .4mm] (0,0) -- (0,1.05) node[left] {1};
		\draw[-, line width = .4mm] (0,0) -- (1.25,0);
		\node at (-0.06,0) {0};
		
\end{tikzpicture}}
{\caption{Generalized individual efficiency analysis of price per unit costs. Each point represents $\equivbid^{\text{PPU}}(\action)$ for some action. Best responses lie along Pareto frontier $\hat \allocation(\cdot)$. This is equivalent to playing against the winner-pays bid mechanism with allocation rule $\hat \allocation(\cdot)$.} \label{fig:pareto}}
\end{SCfigure}
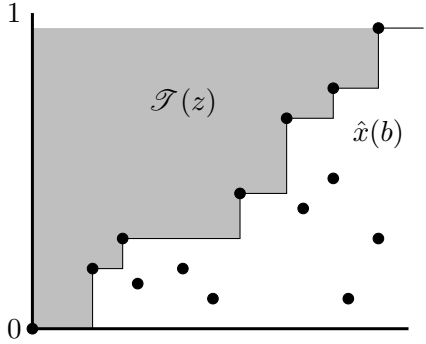

The \vc\ analysis of winner-pays-bid \bos\ in Section~\ref{sec:welf} implies generalized \vc\ guarantees under the price per unit cost function $\equivbid^{\text{PPU}}(\cdot)$.


\begin{lemma}
	The weak individual efficiency of winner-pays-bid mechanisms for price per unit costs under best response is $(e-1)/e$.
\end{lemma}

\begin{proof}
Let $(\val,\action^\info,\bidallocation^\info,\bidpay^\info)$ be a single-agent \ao, and define $\bidallocation(\cdot)=\mathbb E_{\info}[\bidallocation^\info(\cdot)]$.
Then $\bidallocation(\cdot)$ induces a Pareto frontier between allocation and price per unit: ${\wpballoc(\bid)=\sup_{\action\,:\,\equivbid^{\text{PPU}}(\action)\leq \bid}\bidallocation(\action)}$.
See Figure~\ref{fig:pareto} for illustration.
We may treat $\wpballoc(\bid)$ as a winner-pays-bid allocation rule and apply the \vc\ guarantees for winner-pays-bid mechanisms.
Specifically, let $\bid^*$ maximize $(\val-\bid^*)\wpballoc(\bid^*)$.
Then the single-agent \bo\ $(\val,\bid^*,\wpballoc)$ satisfies best response.
Moreover, the winner-pays-bid threshold surplus for $(\val,\bid^*,\wpballoc)$ is equal to the price per unit surplus $\equivcumulative{}{1}{}$ of $\bidallocation(\cdot)$,
 and hence by Lemma~\ref{lem:brres}, $(\val-\bid^*)\wpballoc(\bid^*)+\equivcumulative{}{1}{}\geq \val(e-1)/e$.
 But since $(\val,\action^\info,\bidallocation^\info,\bidpay^\info)$ satisfies best response, $\util(\action^\info,\bidallocation^\info,\bidpay^\info)\geq(\val-\bid^*)\wpballoc(\bid^*)$, yielding the stated bound. 
\end{proof}

We are further able to extend Lemma~\ref{lem:vvalres} to obtain:
\begin{lemma}
	Let $\singlecdf$ be a regular value distribution with monopoly reserve $\reservesingm$, and let $(\val,\action^\info,\bidallocation^\info,\bidpay^\info)$ be a single-agent \ao\ satisfying best response. If $\equivbid^{\text{PPU}}(\action^\info)\geq \reservem$ with probability $1$, then with price per unit costs,
	\begin{equation}\label{eq:genrev}
		\mathbb E_{\info}[\vval(\val)\bidallocation^\info(\action^\info)] + \equivcumulative{}{1}{\reservesingm}\geq\tfrac{e-1}{e}\vval(\val) .\end{equation}
\end{lemma}

Definitions~\ref{def:weakrc} and \ref{def:weakvc} weaken otherwise general definitions of \rc\ and \vc. Nonetheless, they are sufficient to obtain welfare and revenue guarantees:

\begin{theorem}\label{thm:genpoa}
	Let $\mech$ be a mechanism with weak \rc\ $\revpar\leq 1$ for an environment with deterministic allocations.
	Let $\jointcdf$ be an equilibrium in which the agents' conditional \aos\ have individual efficiency $\eta$. Then the expected welfare is an $\revpar\eta$-approximation to that of the optimal mechanism.
\end{theorem}

\begin{theorem}\label{thm:genrev}
	Let $\mech$ be a mechanism with weak \rc\ $\revpar\leq 1$ for an environment with deterministic allocations.
	Let $\jointcdf$ be an equilibrium with product distribution $\valuecdfs$ which satisfies best response and no bidder communication. Further assume that $\jointcdf$ respects reserves: for all $\agent$, $\valagent\geq \reserveagentm$ implies $\equivbid^{\text{PPU}}(\actionagent)\geq \reserveagentm$ and $\valagent< \reserveagentm$ implies $\bidallocagent(\actionagent)=0$ with probability $1$.
	 Then the expected revenue is a $\tfrac{e-1}{2e}\revpar$-approximation to that of the optimal mechanism.
\end{theorem}

We only prove Theorem~\ref{thm:genpoa}, as the proof of Theorem~\ref{thm:genrev} is analogous.

\begin{proof}[Proof of Theorem~\ref{thm:genpoa}]
	Fix a value profile $\vals$, let $\allocoptagent(\vals)$ denote agent $\agent$'s allocation under the welfare-optimal allocation rule. For all agents $\agent$, we may write:
	\begin{equation*}
		\utilagent(\valagent)+\equivcumulative{\agent}{\allocoptagent(\vals)}{}\geq\eta \valagent\allocoptagent(\vals),
	\end{equation*}
	where the inequality follows from weak individual efficiency when $\allocoptagent(\vals)=1$, and trivially when $\allocoptagent(\vals)=0$.
	Summing over all agents and taking expectation over $\vals$ yields:
	\begin{equation*}
		\expect[\vals]{\sum\nolimits_{\agent}\utilagent(\valagent)}+\expect[\vals]{\sum\nolimits_{\agent}\equivcumulative{\agent}{\allocoptagent(\vals)}{}}\geq\eta \expect[\vals]{\sum\nolimits_{\agent}\valagent\allocoptagent(\vals)}.
	\end{equation*}
	The righthand side is the optimal expected welfare. Since for every profile $\vals$, $\allocopt(\vals)$ is feasible, the second term on the lefthand side is at most $\revpar\rev(\mech,\bidcdfs,\valuecdfs)$. Hence:
	\begin{equation*}
		\expect[\vals]{\sum\nolimits_{\agent}\utilagent(\valagent)}+\revpar\rev(\mech,\bidcdfs,\valuecdfs)\geq \eta\WEL(\OPTmech,\valuecdfs).
	\end{equation*}
	The result then follows from noting that $\revpar\leq 1$ and that the welfare of $\mech$ is the sum of the expected utilities and revenue.
\end{proof}

\subsection{Simultaneous Composition}
\label{sec:simultaneous}

\newcommand{\feasibleitem}{\feasible^{j}}
\newcommand{\equivbiditem}{\equivbidagent^{j}}
\newcommand{\equivthreshitem}[2]{\tau_{#1}^j(#2)}
\newcommand{\equivthreshresitem}[3]{\tau_{#1}^{j,#3}(#2)}
\newcommand{\equivcumulativeitem}[3]{{\mathscr T}_{#1}^{j,#3}(#2)}
\newcommand{\withdraw}{\bot}
\newcommand{\indicij}{\delta_{\agent\mitem}}

We now apply our generalized framework and study the impact of local mechanisms' properties on the equilibrium performance of a larger system. Consider $\numagents$ single-parameter agents seeking abstract service. Agents may participate in one or more of $\numitems$ separate mechanisms. Each of the mechanisms are run simultaneously; each agent takes a profile of actions, one per mechanism. Taking allocation levels in $[0,1]$ to be fractional levels of service, we define an agent's service level in the combined mechanism to be their maximum service level across all mechanisms. Agents make the assigned payments to all mechanisms. In this section we show that under these assumptions, if each individual mechanisms has weak \rc\ $\revpar$, then so too does the aggregate mechanism induced for agents by simultaneous participation as described above. We refer to the aggregate mechanism as the {\em simultaneous composition} of the individual component mechanisms.

Formally, a simultaneous composition of mechanisms consists of $\numitems$ separate feasibility environments $\feasible^1,\ldots,\feasible^{\numitems}$, one per component mechanism. Each component mechanism $\mechitem$ is comprised of a bid allocation rule $\bidallocitem$ and a bid payment rule $\bidpaymentitem$, mapping a profile of actions 
$\actionsitem$ to an allocation in $\feasibleitem$ and a nonnegative payment, respectively. We assume each mechanism has a withdraw action $\withdraw$ which guarantees zero allocation and payments in $\mech^{\mitem}$. Define the simultaneous composition of mechanisms $\mech^1,\ldots, \mech^{\numitems}$ in the following way:

\begin{definition}\label{def:simultaneous}
	Let mechanisms $\mech^1,\ldots,\mech^{\numitems}$ have bid allocation and bid payment rules $(\bidallocitem,\bidpaymentitem)$ and individual action spaces spaces $A_\agent^1,\ldots,A_\agent^m$ for each agent $\agent$. The \emph{simultaneous composition} of $\mech^1,\ldots,\mech^m$ is defined to have:
	\begin{itemize}
		\item Action space $\prod_\mitem A_\agent^\mitem$ for each agent. That is, each agent participates in the global mechanism by participating in each composed mechanism individually. Given a profile of actions $\actions$ for each agent, we denote by $\actionsitem$ the profile of actions restricted to mechanism $\mitem$.
		\item Allocation rule $\bidallocagent(\actions)=\max_
		{\mitem} \bidallocagentitem(\actionsitem)$. That is, each agent is served at their highest level across all component mechanisms.
		\item Payment rule $\bidpaymentagent(\actions)=\sum_{\mitem} \bidpaymentagentitem(\actionsitem)$. That is, agents make payments to every composed mechanism.
	\end{itemize}
\end{definition}
Note that as a consequence of Definition~\ref{def:simultaneous}, we may define the composed feasibility environment as the set of allocation levels induced by the component mechanisms, i.e. 
\begin{equation*}
\feasible=\{(\max_{\mitem} \allocation_1^{\mitem},\ldots, \max_{\mitem} \allocation_n^{\mitem})\,|\,\alloc^1,\ldots,\alloc^\numitems\in \feasible^1,\ldots,\feasible^{\numitems}\}.
\end{equation*}

The main result of this section is that weak \rc\ of the component mechanisms implies weak \rc\ for their simultaneous composition.
That is, if each component mechanism efficiently translates competition, measured via threshold costs, into revenue, then the simultaneous composition of those mechanisms will do so as well.

Competitive efficiency holds with respect to cost functions $\equivbidagent$  for the composed mechanism. for each $\agent$, which we may extend to the component mechanisms in the following way.
Let $\action_\agent^\mitem$ be an action for agent $\agent$ in mechanism $\mitem$, and let $\hat \action_\agent^\mitem$ denote the action for agent $\agent$ in the composed mechanism where $\agent$ plays $\actionagentitem$ in mechanism $\mechitem$ and withdraws from all other mechanisms $\mitem'\neq \mitem$.
Let $\actionagent$ be an action in the composed mechanism consisting of actions $\actionagentitem$ in mechanism $\mechitem$ for each $\mitem$. Define $\equivbiditem(\actionagent)=\equivbidagent(\hat\action_\agent^\mitem)$, where we omit conditioning on $\valagent$ due to the independence of $\valuecdfs$. For any allocation $x\in[0,1]$, we may also define agent $\agent$'s threshold cost $\equivthreshitem{\agent}{\allocation}=\inf_{\actionagent:\bidallocagent(\actionagent)\geq x}\equivbiditem(\actionagent)$, and similarly define the threshold cost with discounted reserve $\equivthreshresitem{\agent}{\allocation}{\reserveagent}$ and generalized threshold surplus with discounted reserve $\equivcumulativeitem{\agent}{\allocation}{\reserveagent}$ for mechanism $\mechitem$ analogously to Definitions~\ref{def:reservecost} and \ref{def:cumulativeprice}.
For the example of price per unit costs $\equivbidagent^{\text{PPU}}$, $\equivbidagentitem$ corresponds to the price per unit cost in the component mechanism $\mechitem$ as one would expect.
We may now state the section's main result.

\begin{theorem}
	\label{thm:simultaneous}
	Let $\mech$ be a simultaneous composition of mechanisms $\mech^1, \ldots \mech^m$. If $\mech^1,\ldots, \mech^m$ all have weak \rc\ $\mu$ with minimum costs $\reserves$, then so does $\mech$.
\end{theorem}

The theorem will follow from two observations. First, agents' threshold prices are lower in the composition of mechanisms than in any individual component mechanism. In other words, it is easier for agents to secure allocation with more mechanisms to participate in. Second, the aggregate revenue is the sum of the component mechanisms' revenues.
The following lemma formalizes the first observation.

\begin{lemma}\label{lem:simthresholds}
	For any distribution of actions and values $\jointcdf$, any $z\in[0,1]$, any component mechanism $\mechitem$, and any minimum cost $\reserveagent$, $\equivcumulative{\agent}{z}{\reserveagent}\leq \equivcumulativeitem{\agent}{z}{\reserveagent}$.
\end{lemma}
\begin{proof}
	By definition, $\equivthresh{\agent}{\allocation}=\inf_{\actionagent:\bidallocagent(\actionagent)\geq x}\equivbid(\actionagent)\leq\inf_{\actionagent:\bidallocagent(\actionagent)\geq x}\equivbiditem(\actionagent)=\equivthreshitem{\agent}{\allocation}$, and so $\equivthreshres{\agent}{\allocation}{\reserveagent}\leq\equivthreshresitem{\agent}{\allocation}{\reserveagent}$. Integrating both quantities yields $\equivcumulative{\agent}{\allocation}{\reserveagent}\leq \equivcumulativeitem{\agent}{\allocation}{\reserveagent}$.
\end{proof}

\begin{proof}[Proof of Theorem~\ref{thm:simultaneous}]
	Let $\bidcdfs$ be an action distribution and $\valuecdfs$ a product prior over values, with joint distribution $\jointcdf$. Further let $\allocalt$ be a feasible allocation in the composed environment, where $\allocalt^1\,\ldots,\allocalt^\numitems\in\feasible^1,\ldots,\feasible^{\numitems}$ denotes a profile of allocations for each component mechanism certifying the feasibility of $\allocalt$, i.e. $\allocaltagent=\max_j \allocaltagentitem$ for all $\agent$.
	For each agent $\agent$, let $\indicij$ be an indicator taking value $1$ if $\mitem$ is the lowest index such that $\allocaltagent=\allocaltagentitem$, and $0$ otherwise.
	Further let $\rev^{\mitem}(\mech,\jointcdf)$ denote the revenue from component mechanism $\mechitem$ under $\jointcdf$.
	We obtain the following sequence of inequalities, explained after their statement:
	\begin{align*}
		\tfrac{1}{\revpar}\rev(\mech,\jointcdf)&=\sum\nolimits_{\mitem} \tfrac{1}{\revpar}\rev^{\mitem}(\mech,\jointcdf)\\
		&\geq \sum\nolimits_{\mitem}\sum\nolimits_{\agent}\equivcumulativeitem{\agent}{\allocaltagentitem}{\reserveagent}\\
		&\geq \sum\nolimits_{\mitem}\sum\nolimits_{\agent}\equivcumulative{\agent}{\allocaltagentitem}{\reserveagent}\\
		&=       \sum\nolimits_{\mitem}\sum\nolimits_{\agent}\indicij\equivcumulative{\agent}{\allocaltagent}{\reserveagent}\\
		&=\sum\nolimits_{\agent}\equivcumulative{\agent}{\allocaltagent}{\reserveagent}.
	\end{align*}
	The first inequality follows from the definition of payments in a composed mechanism as the sum of the revenues of the component mechanisms. The second line comes from the assumption that each component mechanism has weak \rc\ $\revpar$ with minimum costs $\reserves$. The third line follows from Lemma~\ref{lem:simthresholds}. The remaining lines follow from the definition of feasibility in the composed mechanism.
\end{proof}

If we consider price per unit costs, then Theorem~\ref{thm:simultaneous} implies that the simultaneous composition of winner-pays-bid mechanisms inherits the \rc\ of the component mechanisms, and hence their robust welfare guarantees.
Moreover, if all mechanisms have monopoly reserves, the same can be said for revenue.

\begin{corollary}
	Let $\mech$ be the simultaneous composition of winner-pays-bid mechanisms, each with competitive efficiency at least $\revpar\leq 1$. Then in any equilibrium with independent values satisfying no bidder communication and best response, $\jointcdf=(\valuecdfs,\bidcdfs)$ with independent $\valuecdfs$, the welfare is at least a $\revpar(e-1)/e$ fraction of optimal.
\end{corollary}

\begin{corollary}
	Let $\mech$ be the simultaneous composition of winner-pays-bid mechanisms, each with competitive efficiency at least $\revpar\leq 1$, and each with monopoly reserves. Then in any equilibrium with independent values, no bidder communication, and best response, and which respects reserves, $\jointcdf=(\valuecdfs,\bidcdfs)$ with independent $\valuecdfs$, the revenue is at least a $\revpar(e-1)/2e$ fraction of optimal.
\end{corollary}

\subsection{All-Pay Mechanisms}
\label{sec:allpay}

This section considers single-bid all-pay mechanisms.
Given an allocation rule $\bidalloc$ mapping bids $\bids=(\bid_1,\ldots,\bid_\numagents)$ to allocations, the corresponding all-pay mechanism has payment rule $\bidpaymentagent(\bids)=\bidagent$.
We will analyze the weak \rc\ of all-pay mechanisms, measuring competition via gross payments rather than price per unit of allocation.
For this cost function, however, individual efficiency fails to hold.
We will instead prove that all-pay mechanisms satisfy the weak \vc\ of Definition~\ref{def:weakvc}.

Given bid distribution $\bidcdfs$, product distribution $\valuecdfs$ over values, and action $\actionagent$ for agent $\agent$, we will consider the cost function $\equivbidagent^{\text{AP}}(\actionagent)=\bidpaymentagent(\actionagent)=\expect[\actionothers]{\bidpaymentagent(\actionagent,\actionothers)}$.
For all-pay mechanisms, payments are equal to bids, so for bid $\bidagent$, we have $\equivbidagent^{\text{AP}}(\bidagent)=\bidagent$.
For such mechanisms, we therefore also have that $\equivthresh{\agent}{\allocation}=\inf_{\bidagent\,:\,\bidallocagent(\bidagent)\geq \allocation} \equivbidagent^{\text{AP}}(\bidagent)$ is the agent's threshold bid for allocation level $\allocation$, and $\equivcumulative{\agent}{1}{}=\int_0^{1}\equivthreshres{\agent}{\allocation}{}\,d\allocation$ is $\agent$'s expected threshold bid.
The connection to threshold bids will enable us to re-use our analysis of winner-pays-bid mechanisms from Section~\ref{sec:rc}.
We present the analysis in the absence of minimum costs $\reserves$, as it is typical to consider all-pay mechanisms without reserves.

\begin{theorem}
	Let $\bidalloc$ be a deterministic allocation rule, and assume the winner-pays-bid mechanism for $\bidalloc$ has \rc\ $\revpar$. Then for cost function $\equivbidagent^{\text{AP}}$, the all-pay mechanism for $\bidalloc$ has weak \rc\ $\revpar$.
\end{theorem}
\begin{proof}
	Let $\bidcdfs$ be a bid distribution and $\valuecdfs$ be a product distribution over values. For a deterministic (ex post) allocation rule $\bidalloc$, let $\mech^{\text{WPB}}$ and $\mech^{\text{AP}}$ be mechanisms obtained by pairing $\bidalloc$ with winner-pays-bid and all-pay payment formats, respectively. Assume $\mech^{\text{WPB}}$ has \rc\ $\revpar$. The result will follow from noting that the all-pay revenue under $\bidcdfs$ and $\valuecdfs$ is weakly larger, while the threshold surplus in $\mech^{\text{WPB}}$ equals the generalized threshold surplus in $\mech^{\text{AP}}$.
	Note that we need only consider $\allocalt\in\{0,1\}^\numagents$ in the threshold surplus, as the environment is deterministic.
	
	The revenue in a winner-pays-bid mechanism is the bids of the winners. The revenue in an all-pay mechanism is the bids of all agents. Therefore for the same bid distribution and prior $\bidcdfs$ and $\valuecdfs$, $\rev(\mech^{\text{AP}},\bidcdfs,\valuecdfs)\geq \rev(\mech^{\text{WPB}},\bidcdfs,\valuecdfs)$. To compare threshold surplus, note that the winner-pays-bid threshold surplus of $(\bidcdfs,\valuecdfs)$ in $\mech^{\text{WPB}}$ is $\cumuthreshagent{\agent}{\allocaltagent}{}=\int_{0}^{\allocaltagent}\interthresh{\agent}{\allocation}\,d\allocation$.
	For $\allocaltagent=1$, this integral is the expected threshold bid for agent $i$.
	Now consider the same $(\bidcdfs,\valuecdfs)$ in $\mech^{\text{AP}}$.
	The generalized threshold surplus for $\equivbidagent^{\text{AP}}$ is $\equivcumulative{\agent}{\allocaltagent}{}=\int_0^{\allocaltagent}\equivthreshres{\agent}{\allocation}{}\,d\allocation$. This is again equal to $\agent$'s expected threshold bid in the all-pay mechanism for $\allocaltagent=1$. Hence,
	$\cumuthreshagent{\agent}{\allocaltagent}{}=\equivcumulative{\agent}{\allocaltagent}{}$.
	We therefore have:
	\begin{equation*}
		\rev(\mech^{\text{AP}},\bidcdfs,\valuecdfs)\geq \rev(\mech^{\text{WPB}},\bidcdfs,\valuecdfs)\geq\revpar\sum\cumuthreshagent{\agent}{\allocaltagent}{}=\revpar\sum\equivcumulative{\agent}{\allocaltagent}{},
	\end{equation*}
where the second inequality follows from the competitive efficiency of $\mech^{\text{WPB}}$.
\end{proof}

We now derive a weak \vc\ guarantee:
\begin{lemma}[Weak \VC\ of All-Pay Mechanisms]\label{lem:weakvc}
	Let $(\val, \bid^\info,\bidallocation^\info,\bidpay^\info)$ be a single-agent \ao\ with all-pay payments: $\bidpay^\info(\bid)=\bid$ for all $\info$.
	If $(\val, \bid^\info,\bidallocation^\info,\bidpay^\info)$ satisfies best response, then with cost function $\equivbidagent^{\text{AP}}$ the following holds: 
	\begin{equation}\label{eq:weakvc}
		\mathbb E_{\info}[\val\bidallocation^\info(\bid^\info)-\bid^\info] + \equivcumulative{}{1}{}\geq\val/2.\end{equation}
\end{lemma}

To see that all-pay mechanisms do not satisfy the stronger definition of \vc, consider the single-bid rule $\bidallocation(b)=b$ and an agent with value $\val=1$. Since $\bidutil(\bid)=0$ for all best responses $\bid$, it follows that $\bidutil(\bid) + \equivcumulative{}{z}{}\rightarrow 0$ as $z^2/2$, which implies that no constant individual efficiency guarantee holds for arbitrary $z$.

\begin{proof}[Proof of Lemma~\ref{lem:weakvc}]
	Let $\bidutil(\bid^\info,\bidallocation^\info,\bidpay^\info)=(\val,\mechdist)$ be given.
	For any allocation probability $\allocation\in[0,1]$, agent $\agent$ could choose to get allocation probability at least $\allocation$ and pay at most $\equivthresh{}{\allocation}$. Hence, $\util(\mechdist)=E_{\info}[\val\bidallocation^\info(\bid^\info)-\bid^\info]\geq \val \allocation-\equivthresh{}{\allocation}$ for all $\allocation\in[0,1]$.
	We may rearrange this as $\equivthresh{}{\allocation}\geq \val\allocation-\util(\mechdist)$. Since we also have $\interthresh{}{\allocation}\geq 0$, we may write:
	\begin{align*}
		\equivcumulative{}{1}{}\geq \int_0^1 \max(\val\allocation-\util(\mechdist),0)\,d\allocation&=\int_{\util(\mechdist)/\val}^1 \val\allocation-\util(\mechdist)\,d\allocation\\
		&=\tfrac{\val}{2}-\util(\mechdist)+\tfrac{\util(\mechdist)^2}{2\val}.
	\end{align*}
	We therefore have $\util(\mechdist)+\equivcumulative{}{1}{}\geq \val/2+\util(\mechdist)^2/2\val$. Holding $\val$ fixed and minimizing the righthand side as a function of $\util(\mechdist)$ yields a lower bound of $\val/2$, as desired. 
\end{proof}

From Theorem~\ref{thm:genpoa}, we may conclude that all winner-pays-bid mechanisms for deterministic environments which have \rc\ $\revpar$ have an analogous all-pay mechanism with the same weak \rc. The welfare in these mechanisms is consequently a $\revpar/2$-fraction of optimal. Notably, this includes the simultaneous composition of all-pay mechanisms.

\begin{corollary}
For a deterministic allocation rule $\bidalloc$, let $\mech^{\text{WPB}}$ and $\mech^{\text{AP}}$ be the respective winner-pays-bid and all-pay mechanisms for $\bidalloc$. Then if $\mech^{\text{WPB}}$ has competitive efficiency $\revpar\leq 1$, then in any equilibrium $(\bidcdfs,\valuecdfs)$ with product distribution $\valuecdfs$ satisfying no bidder communication and best response in $\mech^{\text{AP}}$, the welfare is at least a $\revpar/2$ fraction of optimal.
\end{corollary}

\bibliographystyle{apalike}
\bibliography{bibs}

\begin{thebibliography}{}

\bibitem[Allouah and Besbes, 2020]{AB-20}
Allouah, A. and Besbes, O. (2020).
\newblock Prior-independent optimal auctions.
\newblock {\em Management Science}, 66(10):4417--4432.

\bibitem[Bergemann et~al., 2017]{BBM-17}
Bergemann, D., Brooks, B., and Morris, S. (2017).
\newblock First-price auctions with general information structures:
  Implications for bidding and revenue.
\newblock {\em Econometrica}, 85(1):107--143.

\bibitem[Bergemann et~al., 2019]{BBM-19}
Bergemann, D., Brooks, B., and Morris, S. (2019).
\newblock Revenue guarantee equivalence.
\newblock {\em American Economic Review}, 109(5):1911--29.

\bibitem[Bergemann and Morris, 2016]{BM16}
Bergemann, D. and Morris, S. (2016).
\newblock Bayes correlated equilibrium and the comparison of information
  structures in games.
\newblock {\em Theoretical Economics}, 11(2):487--522.

\bibitem[Brooks and Du, 2021]{BD-21}
Brooks, B. and Du, S. (2021).
\newblock Optimal auction design with common values: An informationally robust
  approach.
\newblock {\em Econometrica}, 89(3):1313--1360.

\bibitem[Carroll, 2017]{car-17}
Carroll, G. (2017).
\newblock Robustness and separation in multidimensional screening.
\newblock {\em Econometrica}, 85(2):453--488.

\bibitem[Devanur et~al., 2015]{DHY15}
Devanur, N.~R., Hartline, J.~D., and Yan, Q. (2015).
\newblock Envy freedom and prior-free mechanism design.
\newblock {\em Journal of Economic Theory}, 156:103--143.

\bibitem[Dhangwatnotai et~al., 2010]{DRY10}
Dhangwatnotai, P., Roughgarden, T., and Yan, Q. (2010).
\newblock Revenue maximization with a single sample.
\newblock In {\em ACM Conference on Electronic Commerce}, pages 129--138.

\bibitem[Dhangwatnotai et~al., 2015]{DRY-15}
Dhangwatnotai, P., Roughgarden, T., and Yan, Q. (2015).
\newblock Revenue maximization with a single sample.
\newblock {\em Games and Economic Behavior}, 91:318--333.

\bibitem[D\"utting and Kesselheim, 2015]{DK15}
D\"utting, P. and Kesselheim, T. (2015).
\newblock Algorithms against anarchy: Understanding non-truthful mechanisms.
\newblock In {\em 16th ACM Conference on Economics and Computation}.

\bibitem[Edelman et~al., 2007a]{EOS-07}
Edelman, B., Ostrovsky, M., and Schwarz, M. (2007a).
\newblock Internet advertising and the generalized second-price auction:
  Selling billions of dollars worth of keywords.
\newblock {\em American economic review}, 97(1):242--259.

\bibitem[Edelman et~al., 2007b]{EOS07}
Edelman, B., Ostrovsky, M., and Schwarz, M. (2007b).
\newblock Internet advertising and the generalized second-price auction:
  Selling billions of dollars worth of keywords.
\newblock {\em American economic review}, 97(1):242--259.

\bibitem[Feldman et~al., 2013]{FFGL13}
Feldman, M., Fu, H., Gravin, N., and Lucier, B. (2013).
\newblock Simultaneous auctions are (almost) efficient.
\newblock In {\em Proceedings of the forty-fifth annual ACM symposium on Theory
  of computing}, pages 201--210.

\bibitem[Forges, 1986]{F86}
Forges, F. (1986).
\newblock An approach to communication equilibria.
\newblock {\em Econometrica: Journal of the Econometric Society}, pages
  1375--1385.

\bibitem[Fu et~al., 2015]{FILS-15}
Fu, H., Immorlica, N., Lucier, B., and Strack, P. (2015).
\newblock Randomization beats second price as a prior-independent auction.
\newblock In {\em Proceedings of the Sixteenth ACM Conference on Economics and
  Computation}, pages 323--323.

\bibitem[Hartline et~al., 2014]{HHT-14}
Hartline, J., Hoy, D., and Taggart, S. (2014).
\newblock Price of anarchy for auction revenue.
\newblock In {\em Proceedings of the fifteenth ACM conference on Economics and
  computation}, pages 693--710.

\bibitem[Hartline et~al., 2020]{HJL-20}
Hartline, J., Johnsen, A., and Li, Y. (2020).
\newblock Benchmark design and prior-independent optimization.
\newblock In {\em 2020 IEEE 61st Annual Symposium on Foundations of Computer
  Science (FOCS)}, pages 294--305. IEEE.

\bibitem[Hartline et~al., 2015]{HST15}
Hartline, J., Syrgkanis, V., and Tardos, E. (2015).
\newblock No-regret learning in bayesian games.
\newblock {\em Advances in Neural Information Processing Systems}, 28.

\bibitem[Hartline, 2013]{har-13}
Hartline, J.~D. (2013).
\newblock Mechanism design and approximation.
\newblock {\em Book draft. October}, 122:1.

\bibitem[Hartline and Roughgarden, 2009]{HR09}
Hartline, J.~D. and Roughgarden, T. (2009).
\newblock Simple versus optimal mechanisms.
\newblock In {\em ACM Conference on Electronic Commerce}, pages 225--234.

\bibitem[Hoy et~al., 2015]{HNS15}
Hoy, D., Nekipelov, D., and Syrgkanis, V. (2015).
\newblock Robust data-driven efficiency guarantees in auctions.
\newblock In {\em EC Workshop on Algorithmic Game Theory and Data Science}.

\bibitem[Jin and Lu, 2022]{JL22}
Jin, Y. and Lu, P. (2022).
\newblock First price auction is 1--1/e2 efficient.
\newblock In {\em 2022 IEEE 63rd Annual Symposium on Foundations of Computer
  Science (FOCS)}, pages 179--187. IEEE Computer Society.

\bibitem[Kaplan and Zamir, 2012]{KZ12}
Kaplan, T.~R. and Zamir, S. (2012).
\newblock Asymmetric first-price auctions with uniform distributions: analytic
  solutions to the general case.
\newblock {\em Economic Theory}, 50(2):269--302.

\bibitem[Lehmann et~al., 2002]{LOS02}
Lehmann, D., O{\'c}allaghan, L.~I., and Shoham, Y. (2002).
\newblock Truth revelation in approximately efficient combinatorial auctions.
\newblock {\em Journal of the ACM (JACM)}, 49(5):577--602.

\bibitem[Myerson, 1981]{M81}
Myerson, R. (1981).
\newblock Optimal auction design.
\newblock {\em Mathematics of Operations Research}, 6(1):58--73.

\bibitem[Roughgarden, 2009]{R09}
Roughgarden, T. (2009).
\newblock Intrinsic robustness of the price of anarchy.
\newblock In {\em ACM Symposium on Theory of Computing}, pages 513--522.

\bibitem[Roughgarden et~al., 2012]{RTY12}
Roughgarden, T., Talgam-Cohen, I., and Yan, Q. (2012).
\newblock Supply-limiting mechanisms.
\newblock In {\em ACM Conference on Electronic Commerce}, pages 844--861.

\bibitem[Syrgkanis, 2014]{S14}
Syrgkanis, V. (2014).
\newblock {\em Efficiency of mechanisms in complex markets}.
\newblock PhD thesis, Cornell University.

\bibitem[Syrgkanis and Tardos, 2013]{ST13}
Syrgkanis, V. and Tardos, E. (2013).
\newblock Composable and efficient mechanisms.
\newblock In {\em ACM Symposium on Theory of Computing}, pages 211--220.

\bibitem[Varian, 2007]{var-07}
Varian, H.~R. (2007).
\newblock Position auctions.
\newblock {\em international Journal of industrial Organization},
  25(6):1163--1178.

\end{thebibliography}

\appendix

\end{document}